\documentclass{elsarticle} 
\usepackage[a4paper, total={6.5in, 9in}]{geometry}

\makeatletter
\def\ps@pprintTitle{%
 \let\@oddhead\@empty
 \let\@evenhead\@empty
 \def\@oddfoot{\centerline{\thepage}}%
 \let\@evenfoot\@oddfoot}
\makeatother
\usepackage{balance} % for balancing columns on the final page

\usepackage{etoolbox}
\usepackage{xspace}
\usepackage{fancyhdr}
\usepackage{tcolorbox}

\usepackage{amssymb}

\usepackage{amsthm}
\usepackage{amsmath}
\usepackage{hyperref}
\usepackage{cleveref}

\newtheorem{theorem}{Theorem}
\newtheorem{lemma}{Lemma}
\newtheorem{corollary}{Corollary}

\newcommand{\NP}{\ensuremath{\mathtt{NP}}\xspace}
\newcommand{\PSPACE}{\ensuremath{\mathtt{PSPACE}}\xspace}

\newcommand{\Wtwo}{\ensuremath{\mathtt{W}[2]}\xspace}

\newcommand{\scal}{\ensuremath{\mathcal{S}}\xspace}

\newcommand{\gcal}{\ensuremath{\mathcal{G}}\xspace}
\newcommand{\ecal}{\ensuremath{\mathcal{E}}\xspace}
\newcommand{\tcal}{\ensuremath{\mathcal{T}}\xspace}

\newcommand{\tuple}[1]{\ensuremath{\langle {#1} \rangle}\xspace}

\newcommand{\tmax}{\ensuremath{t_{\max}}\xspace}
\newcommand{\activeVertex}{\ensuremath{\mathtt{active}}\xspace}

\newcommand{\activeCounter}{\ensuremath{\mathtt{counter}}\xspace}

\newcommand{\problemone}{\textsc{MaxSpread}\xspace}
\newcommand{\problemtwo}{\textsc{MaxViral}\xspace}
\newcommand{\problemthree}{\textsc{MaxViralTstep}\xspace}
\newcommand{\problemfour}{\textsc{MinNonViralTime}\xspace}

\newcommand{\poly}{\ensuremath{\mathtt{poly}}\xspace}

\newcommand{\setcover}{\textsc{SetCover}\xspace}

\newcounter{mycounter} % create a new counter, called 'mycounter'
\newcommand\showmycounter{\stepcounter{mycounter}\themycounter}

\newif\iflong
\newif\ifshort

\longtrue

\iflong
\else
\shorttrue
\fi

\begin{document}

\begin{frontmatter}



\title{Being an Influencer is Hard: The Complexity of Influence Maximization in Temporal Graphs with a Fixed Source}

%\tnotetext[t1]{A preliminary version of the results in this paper has appeared in .}


\author[1]{Argyrios Deligkas}
\ead{argyrios.deligkas@rhul.ac.uk}
\author[2]{Michelle D\"oring}
\ead{michelle.doering@hpi.de}
\author[1]{Eduard Eiben}
\ead{eduard.eiben@rhul.ac.uk}
\author[1]{Tiger-Lily Goldsmith}
\ead{tigerlily.goldsmith@gmail.com}
\author[2]{George Skretas}
\ead{georgios.skretas@hpi.de}

\address[1]{Royal Holloway, University of London, Egham, United Kingdom}
\address[2]{Hasso Plattner Institute, University of Potsdam, Potsdam, Germany}

%\cortext[cor1]{Corresponding author (Telephone number: +49 1786686368, Postal Address: Hasso Plattner Institute, University of Potsdam, Prof.-Dr.-Helmert-Str. 2-3, 14482 Potsdam, Germany)}

% \affiliation{organization={},%Department and Organization

\begin{abstract}
We consider the influence maximization problem over a temporal graph, where there is a single fixed source.
We deviate from the standard model of influence maximization, where the goal is to choose the set of most influential vertices.
Instead, in our model we are given a fixed vertex, or source, and the goal is to find the best time steps to transmit so that the influence of this vertex is maximized.
We frame this problem as a spreading process that follows a variant of the susceptible-infected-susceptible (SIS) model and we focus on four objective functions.
In the \problemone objective, the goal is to maximize the total number of vertices that get infected at least once.
In the \problemtwo objective, the goal is to maximize the number of vertices that are infected at the same time step.
In the \problemthree objective, the goal is to maximize the number of vertices that are infected at a given time step.
Finally, in \problemfour, the goal is to maximize the total number of vertices that get infected every $d$ time steps.
We perform a thorough complexity theoretic analysis for these four objectives over three different scenarios: (1) the unconstrained setting where the source can transmit whenever it wants; (2) the window-constrained setting where the source has to transmit at either a predetermined, or a shifting window; (3) the periodic setting where the temporal graph has a small period.
We prove that all of these problems, with the exception of \problemone for periodic graphs, are intractable even for very simple underlying graphs. %%%%%%%%%%%%%%%%%%%%%%%%%%%%%%%%%%%%%%%%%%%%%%%%%%%

\end{abstract}


% \begin{highlights}
% \end{highlights}

\begin{keyword}
Influence Maximization; Social Networks; Temporal Graphs; Computational Complexity; Parameterized Complexity

\end{keyword}

\end{frontmatter}

\newcommand{\BibTeX}{\rm B\kern-.05em{\sc i\kern-.025em b}\kern-.08em\TeX}

\section{Introduction}
{\em ``When is the right time to post our content online? Which days should we place our advertisements so our video goes viral? What's the best strategy to maximize the word-of-mouth effect?''} These are some of the questions advertisers, political parties, and individual influencers want to answer. In every case, the goal is the same: maximize their {\em influence} over their social network.

By now, influence maximization is a well-established problem in Computer Science. The seminal paper of \cite{kempe2003maximizing} introduced the basic mathematical model of influence maximization and became the foundation for a plethora of follow-up  models. There, the input consists of a static network and the task is to find a set of initial spreaders, or sources, that maximizes the expected outbreak size of a spreading process occurring over the network. 

%In real-world networks though, many situations and applications exhibit an inherited {\em temporal structure}: a person will check their social media just a few times during their day; a user will go through their favorite webpages every couple of days; people meet with their friends every few days. Furthermore, the effectiveness of word-of-mouth for a specific post/product/advert deteriorates as time passes if there are no interactions, or discussions, about this. Arguably, if a person has not recently seen a post from their favourite artist, they are unlikely to discuss with their friends a post the artist made a long time ago. Hence, the base model of~\cite{kempe2003maximizing} has to be suitably augmented in order to capture the above mentioned scenarios.

In real-world networks though, many situations and applications exhibit an inherent {\em temporal structure}: a person will check their social media just a few times during their day; a user will go through their favorite web pages every couple of days; people meet with their friends every few days. Furthermore, the effectiveness of word-of-mouth for a specific post/product/advert deteriorates as time passes if there are no interactions, or discussions, about it. Arguably, if a person has not recently seen a post from their favourite artist, it is highly unlikely that the artist will become a topic of conversation with their friends.
Motivated by the increase of social media influence in marketing strategies and the belief that the data provided will become more detailed in the future, the base model of~\cite{kempe2003maximizing} has to be suitably augmented in order to capture the above mentioned scenarios.

Temporal graphs~\cite{holme2012temporal,kempe02-temporal} form a solid basis that naturally captures the temporal structure of the aforementioned instances;
in a temporal graph there is a fixed set of vertices and the connections between them change between consecutive time steps.
Furthermore, when a temporal graph is coupled with a susceptible-infected-susceptible (SIS) spreading process, then it becomes an excellent framework to analyze the influence maximization problems described previously. 
%That is why, this model has been adopted by a great number of works~\argy{cite}. 
So, now the task is not only to choose the vertices that will become the sources, but in addition decide {\em when}, i.e., at which time step, each source will become active~\cite{gayraud2015diffusion}. 
In fact, marketing companies are very interested about the {\em ``when-to-post"} problem, since different times of posting new content can yield vastly different results. This problem has already been studied both within computer science \cite{spasojevic2015post,zarezade2017redqueen} and in the field of marketing \cite{kanuri2018scheduling}.

Since the influence maximization problem on temporal graphs is a generalization of the base model of~\cite{kempe2003maximizing}, the \NP-hardness results derived for the more constrained setting immediately apply here too. 
On the other hand, all the models so far assume that we can choose each vertex only {\em once}~\cite{gayraud2015diffusion}. Thus, the implied intractability of the problem comes from the choice of the sources and not from the activation times. 
In reality though, the set of possible sources someone can choose from is rather limited and not the whole set of vertices of the graph. In fact, there are several cases where there is only a single, {\em fixed} source; for example ``youtubers'' are independent and the only power they have is just to choose the time they will release their videos online. In this case, the \NP-hardness from~\cite{kempe2003maximizing} does not apply any more. Our goal is to remedy this situation by establishing the tractability frontier for this scenario.

%%%%%%%%%%%%%%%%%%%%%%%%%%%%%%%%%%%%%%%%%%
\subsection{Our contribution}
Our contribution is two-fold. Firstly, we formally define 
\ifshort three \fi \iflong four \fi different influence-maximization objectives on temporal graphs with a fixed source, where each variant captures a different aspect of the problem. Then, we perform a complexity-theoretic analysis for them, under three different variants that arise in real-life scenarios.

The input in each problem consists of a temporal graph, known in advance, a fixed source, and a budget of allowed ``posts'', or {\em transmissions}. The goal is to choose a {\em transmission schedule}, i.e., the time steps the source transmits, in order to maximize an influence-related objective. 
We follow an SIS-style spreading process: an {\em inactive}, i.e., not currently influenced, vertex becomes {\em active} at time step $t+1$ if at time step $t$, it is adjacent to an active vertex. Every active vertex is associated with a {\em counter} which shows for how many time steps it will remain active. If at time step $t$ a vertex is adjacent to an active vertex, then its counter resets to a fixed number $\delta$. The parameter $\delta$ resembles the maximum time an agent will be under the influence of the source, and thus the vertex will diffuse the information it has about the source.

The first objective, termed \problemone, aims to maximize the number of different vertices that have become influenced at some point in time. 
This is probably the most natural objective, since it aims to maximize the exposure of the source to different customers.
The goal of the second objective, termed \problemtwo, is to create a ``viral effect'': maximize the number of simultaneously-active vertices.
\ifshort Finally, the target of the objective \problemthree is to maximize the number of active vertices at a predetermined time step. For example, this objective is desirable when a political party wants to maximize the number of active voters on elections day, or a company wants to maximize its active customers on a product-launch day. \fi
\iflong Next, the target of the objective \problemthree is to maximize the number of active vertices at a predetermined time step. For example, this objective is desirable when a political party wants to maximize the number of active voters on elections day, or a company wants to maximize its active customers on a product-launch day. 
Finally, the goal of the objective \problemfour, is to maximize the number of different vertices that became influenced at some point in time, while ensuring that no vertex is inactive for more than a predefined amount of time steps. This objective is appropriate for established products in the market that want to remain "fresh" in customer's minds. 
\fi

The first scenario we study is when there are no constraints on the transmission schedule, i.e., the source can transmit at any time step. In the second scenario, we consider transmission schedules that have to follow {\em window constraints}, a type of constraints that actually occur in real life. 
%In addition, every marketing strategy includes social media advertising and specifically, companies pick "influencers" to become their brand ambassadors and advertise their product on their social media account. 
When a company chooses an influencer, they specify strict times and dates where each individual post has to be done. The timing mainly depends on abstract analytics that specify the influencer's highest engagement in the day but the data currently provided are not pinpoint and they are bundled in 3-hour slots.
We consider two cases of window constraints. In the {\em fixed} window case, the time is split into intervals of length $w$ and the source has to transmit exactly once in each window. In the $(x,y)$-{\em shifting} window case, any two consecutive transmissions of the source have to be at least $x$ time steps apart and at most $y$ time steps apart.
In the last scenario, we consider {\em periodic} graphs. In this case the temporal graph has infinite lifetime with period $\tmax$. This means that the edges that are available at time step $i \in [1, \tmax]$, appear at time step $i + j \cdot \tmax$, for every $j\in \mathbb{N}$.

\medskip
\noindent {\bf Our results.}
We perform a thorough complexity-theoretic study for the four objectives in each of the aforementioned scenarios. 
%Table~\ref{tab:results} provides an overview of our results\argy{fill the table}.
With the exception of \problemone and \problemfour for periodic graphs, we prove that the %remaining 
problems are intractable even for very restricted settings! In particular, we prove the following results. 
For the unconstrained setting, all four problems are \NP-complete and \Wtwo-hard when parameterized by the number of transmissions of the source, even when the underlying graph is a tree with maximum degree $4$.
For the fixed window setting, the first three problems are \NP-complete even when the window has size two. 
In addition, for \problemone the underlying graph is a star, at most three edges appear at any time step, and every edge appears at most twice.
For \problemtwo and \problemthree the underlying graph is a subdivision of a star. 
%
%In addition, we prove that all problems are \Wtwo-hard when parameterized by the number of windows. \argy{do we actually write this?}\eduard{Just in text we explain how to extend hardness from Section 3 to Section 4 with same $b$ if we do not restrict window size, but do not say it is \Wtwo hard by $b$ explicitly. }
For the shifting window setting, we prove that the first three problems are \NP-complete for every $(2\delta,4\delta)$-shifting window and the underlying graph is a star, or a subdivision of a star.
%with $1<a$ and $b \geq a+2$ \argy{check b}\ even when the underlying graph is a tree of maximum degree \argy{fix}\eduard{again star/subdivision of a star. I was lazy to write it this general, because I would actually have to reprove some things, so just wrote $(2\delta, 4\delta)$ for every $\delta$}. 
Finally, for periodic graphs our results are as follows. For \problemone, we derive a fixed parameter tractable algorithm parameterized by $\tmax$; in other words, the problem is polynomial-time solvable when the period is constant. For problems \problemtwo and \problemthree, we prove that they are \NP-hard even when $\tmax = 2$, i.e., the graph has period 2.

% \begin{table}[]
% \centering
% %\resizebox{\columnwidth}{!}{%
% \begin{tabular}{|l||c|c|c|}
% \hline
%  & \problemone & \problemtwo & \problemthree  \\
%  \hline \hline
%  Unconstrained & \NP-c, \Wtwo-h & \NP-c, \Wtwo-h &  \\
%   & Thm:\ref{thm:spread-uncon1}& Thm:\ref{thm:spread-uncon2} &  \\
%  \hline
%  Fixed Window &  &  &   \\
%  \hline
%  Shifting Window &  &  &  \\
%  \hline
%  Periodic &  &  &  \\
%  \hline
% \end{tabular}
% %}
% \caption{Overview of our results. \argy{add details about the parameters?}}
% \label{tab:results}
% \end{table}

\subsection{Related Work}
After the seminal paper of~\cite{kempe2003maximizing}, influence maximization problems have received a tremendous amount of attention; see for example the surveys~\cite{li2018influence-survey,arora2017debunking-survey2,chen2013information-survey11,guille2013information-survey40,sun2011-survey96,tejaswi2016diffusion-survey101,zhang2014recent-survey113} and the references therein. More related, but still quite different, to our model are {\em time-aware} diffusion models, that study both discrete-time settings~\cite{chen2012time-temp14,kim2014ct-temp60,lin2015learning-temp71}, or continuous-time models~\cite{xie2015dynadiffuse-temp111,rodriguez2011uncovering-temp90,scaman2015anytime-temp93}.

% problem in a social network first introduced as an algorithmic problem \cite{kempe2003maximizing}, proves optimization of select subset of most influential nodes is NP-hard and gave approximate algorithm (greedy solution).
% Using the two diffusion based models: linear threshold \cite{granovetter1978threshold}
% independent cascade model previously investigated in the context of viral marketing \cite{goldenberg2001talk} and \cite{goldenberg2001using}.
Closer to our work are the papers that studied networks that change over time - these networks/graphs can be either temporal, like ours, or graphs that evolve according to different procedures. 
%In previous work the focus was finding the seed set (the most influential subset of nodes) in a dynamically changing influence graph. 
\cite{aggarwal2012influential} propose efficient heuristics for finding a set of sources that maximises influence at a later time step. 
\cite{zhang2013maximizing} assume that the temporal graph is unknown and changes can only be detected periodically.
%; for this setting they provide two new algorithms for constructing a subgraph and finding a set $S$ on the observed sub graph. 
\cite{zhu2014maximizing} introduce a continuous-time Markov chain model, an information diffusion model based on the independent cascade model \cite{goldenberg2001talk,goldenberg2001using} and their experiments illustrate how to find a small set of sources. 
\cite{wang2018modeling} focus on tracking the diffusion and aggregation of influence through the network in the context of viral marketing. They provide heuristics for finding a good initial set of influencers (referred to in their work as adopters).
Another line of research considers location-aware influence maximisation, where the diffusion is focused on vertices located in certain areas rather than a general spread. \cite{zhou2015location} focus on a distance-aware weighted model and go beyond generic influence maximisation algorithms, providing a location-based influence maximisation algorithm.

Recently, a new line of work established to capture the nature of real world social networks moving from static graphs to temporal graphs \cite{ijcai2021p7}. They study the complexity of competitive diffusion games~\cite{alon2010note,fukuzono2020two} and Voronoi games~\cite{durr2007nash,ahn2004competitive}.

In more recent work, \cite{erkol2020influence} studies influence maximisation under the susceptible-infected-recovered (SIR) model and analyses the performance of approximation algorithms over several temporal networks; in addition they study a special temporal network setting where the influence function is not submodular.
In~\cite{erkol2022effective}, the same set of authors answer some of their open questions and they find that greedy optimization is an effective method for finding a set of sources that has very high performance.

%\cite{zhang2013maximizing}, give a different take on the model and take into account agents opinion, hence they study an opinion maximisation problem which they show to be NP-hard. They aim not only to raise awareness and spread influence but specifically spread "positive" influence.

\section{Preliminaries}
\label{sec:prelims}
For $n\in \mathbb{N}$, we denote $[n]:= \{1,2, \ldots, n\}$.
A {\em temporal graph} $\gcal := \tuple{G, \ecal}$ is defined by an {\em underlying graph} $G=(V,E)$ and a sequence of edge-sets $\ecal = (E_1, E_2, \ldots, E_{\tmax})$. It holds that $E = E_1\cup E_2\cup \dots \cup E_{t_{max}}$. The {\em lifetime} of $\gcal$ is \tmax.
An edge $e \in E$ has label $i$, if it is available at time step $i$, i.e., $e \in E_i$. In addition, an edge has $k$ {\em labels}, if it appears in $k$ edge-sets.
%A {\em temporal walk} in $\tuple{G, \ecal}$ from vertex $v_1$ to vertex $v_w$ is a sequence of edges $W = (v_iv_{i+1}, t_i)_{i=1}^{w-1}$ such that for every $i \in [w]$ it holds that $v_iv_{i+1} \in E_{t_i}$, i.e. $v_iv_{i+1}$ is available at time step $t_i$ and time steps are strictly increasing, i.e. if $i < j$ then $t_i < t_j$. Throughout the paper, we assume that every edge has traversal time $1$. Observe that a vertex $v$ can appear multiple times in a temporal walk. A {\em temporal walk} $W = (v_iv_{i+1}, t_i)_{i=1}^{w-1}$ is called $\delta$-restless if $t_i<t_{i+1}<t_i+\delta$, for every $i \in [w]$.

%\paragraph{\bf Spreading process.}
\medskip
\noindent {\bf Spreading process.}
We follow a spreading process that resembles the {\em Susceptible Infected Susceptible} (SIS) model. In our model, we are given a fixed vertex $s \in V$, called the {\em source}. 
In addition, at every time step each vertex has a {\em state}, which is {\em active} or {\em inactive} and it is determined by a {\em counter} that gets integer values in $[0,\delta]$, for some given natural $\delta > 0$. The state of a vertex is active if $\delta > 0$ and inactive otherwise.
%Formally, at time step $t$ 
Initially, every vertex is inactive. 
The source can choose which time steps to {\em transmit}. If $s$ decides to transmit at time step $t$, then it sets its counter to $\delta$, i.e. becomes active and it remains active until time step $t+\delta$. 
The remaining vertices evolve as follows.
\begin{itemize}
    \item If at time step $t$ there is a vertex $v \neq s$ that has an edge with an active vertex $u$, i.e., $uv \in E_t$, then at time step $t+1$ the counter of $v$ is set to $\delta$.
    \item If at time step $t$ there is a vertex $v$ that is active and all of its adjacent vertices are inactive, i.e., all $u \in V$ with $uv \in E_t$ are inactive, then at time step $t+1$ the counter of vertex $v$ decreases by 1.
\end{itemize}
Observe that the procedure above allows ``renewal'' of the active state for a vertex. In other words, it resets the counter of a vertex to $\delta$ after any time step it is adjacent to an active vertex. In contrast, the standard SIS model does not allow renewals. We discuss the differences between our model and SIS in Section~\ref{sec:discuss}.

%\paragraph{\bf Transmission schedules.}
% \medskip
% \noindent {\bf Transmission schedules.}
A {\em transmission schedule} $\tcal= (\tau_1, \tau_2, \ldots, \tau_b)$ is a set of time steps where the source transmits. Let $(\delta,\tcal)-\activeVertex_t(\gcal,s)$ denote the set of active vertices at time step $t$ under the transmission schedule $\tcal$ for source $s$, when the counter is $\delta$.
Furthermore, for any vertex $v\in V$ let $A_v = [t \colon v\in (\delta,\tcal)-\activeVertex_t(\gcal,s)]$ be an ordered list of the time steps at which vertex $v$ is active.

\medskip
\noindent {\bf Problems.}
We study \ifshort three \fi \iflong four \fi problems. In every case, the input is a temporal graph $\gcal$, a source $s$, a budget $b$ and positive integers $\delta$, and $k$.
\begin{itemize}
    \item {\bf \problemone.} Is there a transmission schedule $\tcal$, such that $\left|\bigcup_{t\in[\tmax]}(\delta,\tcal)-\activeVertex_t(\gcal,s)\right| \geq k$ and $|\tcal| \leq b$?
    Put simply, the goal for \problemone is to maximize the number of different vertices that become active.
    \item {\bf \problemtwo.} Is there a transmission schedule $\tcal$, such that $\max_{t \in [\tmax]}$ $\left| (\delta,\tcal)-\activeVertex_t(\gcal,s)\right| \geq k$ and $|\tcal| \leq b$?
    So, the goal for \problemtwo is to maximize the number of active vertices at any time step.
    \item {\bf \problemthree.} Here, we have in addition a time step $t^* \in [\tmax]$.  Is there a transmission schedule $\tcal$, such that $\left|(\delta,\tcal)-\activeVertex_{t^*}(\gcal,s)\right| \geq k$ and $|\tcal| \leq b$?
    In other words, the goal for \problemthree is to maximize the number of active vertices at a given time step $t^*$.
    \iflong
    \item {\bf \problemfour.} Here, we additionally have a positive integer $d$.
    Is there a transmission schedule $\tcal$, such that $\left|\bigcup_{t\in[\tmax]}(\delta,\tcal)-\activeVertex_t(\gcal,s)\right| \geq k$, $|\tcal| \leq b$, and $\max_{0\leq i<\lvert A_v\rvert}\lvert A_v[i+1]-A_v[i]\rvert \leq d$ for all vertices $v\in V$?
    Thus, the goal for \problemfour is to maximize the number of active vertices at any time step, while ensuring that no vertex is inactive for more than $d$ time steps after it is activated for the first time. \fi %TODO
\end{itemize}

\noindent{\textbf{Parameterized complexity.}} 
%We refer to the handbook by Diestel~\shortcite{Diestel12} for
%standard graph terminology. 
We refer to the standard books for a basic overview of parameterized complexity theory~\cite{CyganFKLMPPS15,DowneyFellows13}. At a high level, parameterized complexity studies the complexity of a problem with respect to its input size, $n$,  and the size of a parameter $k$. A problem is {\em fixed parameter tractable} by $k$, if it can be solved in time $f(k)\cdot \poly(n)$, where $f$ is a computable function. If a problem is \Wtwo-hard with respect to $k$, then it is unlikely to be fixed parameter tractable by $k$.

\noindent{\bf \setcover.} 
An instance of \setcover consists of a collection $S=\{S_1, S_2, \ldots, S_m\}$ of subsets of a set $N$ of $n$ elements, and a positive integer $b$. The task is to decide if there are $b$ sets $T_1, T_2, \ldots, T_b$ in $S$, such that $T_1 \cup T_2 \cup \cdots \cup T_b = N$.
%is a $T \subset S$ of size $b$, such that $T \cap S = N$. 
\setcover is known to be \Wtwo-hard when parameterized by $b$~\cite{downey1995fixed}.

%%%%%%%%%%%%%%%%%%%%%%%%%%%%%%%%%%%%%%%%%%%%%%%%%%%%%%%%%%
\section{Unconstrained Schedules}
\label{sec:unconstrained}

In this section we study the complexity of the three objectives under unconstrained transmission schedules. We show that all of them are intractable, even on trees of maximum degree three. For clarity of exposition, when we write that ``vertex $u$ influences vertex $y$'' we mean that vertex $y$ became active via a sequence of vertex activations that includes vertex $s$.

\subsection{\problemone}
We prove that \problemone is \NP-complete and \Wtwo-hard when parameterized by $b$. Containment in \NP is straightforward, since given a transmission schedule, we can simulate the process and check whether every vertex becomes active at least once.
Now, we describe the construction and prove some key properties our construction satisfies. Then we provide the proof of the main theorem of the section.

% \george{Missing np-completeness for obj1,2,3 and w[2] for objective 3}
% We will show that \problemone is \Wtwo-hard when~parameterized~by~$b$ on tree graphs with degree $3$. Before doing that, we will show the construction that we will use for the proof and we will prove some key properties needed. The proof will be via a reduction from \setcover. W.l.o.g. we will assume that $n=2^k$. We also remind to the reader that a perfect binary tree is a tree graph where every internal vertex has $2$ children and the leaf nodes are at the same depth.

\medskip
\noindent {\bf Construction.} 
We reduce from \setcover. In what follows we will assume, that $n=2^k$ for some positive integer $k$. Observe that this is without loss of generality, since we can augment any instance by adding a dummy set that contains the required number of elements, and by asking for a solution of size $b+1$.
We construct a perfect binary tree $P_1$ with $n-1$ leaves, where the root of the binary tree is the source, $s$, and its leaves are  $x_1,x_2,\ldots,x_n$ from ``left to right" and let $h$ be the current height of the tree. 
Note that such a tree always has $2n$ vertices and is uniquely defined. 

Let set $C$ contain every vertex of the current graph (this will be useful later). 

For every $j\in [m]$ and for every $l\in [h]$, we add label $(l-1)\delta+(j-1)(\delta+1)+1$ to every edge of the tree between levels $l$ and $l+1$. Then we construct vertices $y_1,y_2,\ldots, y_n$ and for every $i\in [n]$, we add edge $x_iy_i$. Note now that $P_1$ has $h+1$ levels. For every $i\in [n]$ and $j\in [m]$, if $i\in S_j$, we add label $h\delta+(j-1)(\delta+1)+1$ to edge $x_iy_i$; see Figure~\ref{fig:objective1}.% for an example with $\delta=3$ and $S=\{S_1, S_2,S_3\}$, where $S_1=\{u_1,u_2,u_3\}$, $S_2=\{u_3,u_7,u_8\}$, $S_3=\{u_4,u_5,u_6\}$.

\begin{figure}
    \centering
    \includegraphics[width=0.70\textwidth]{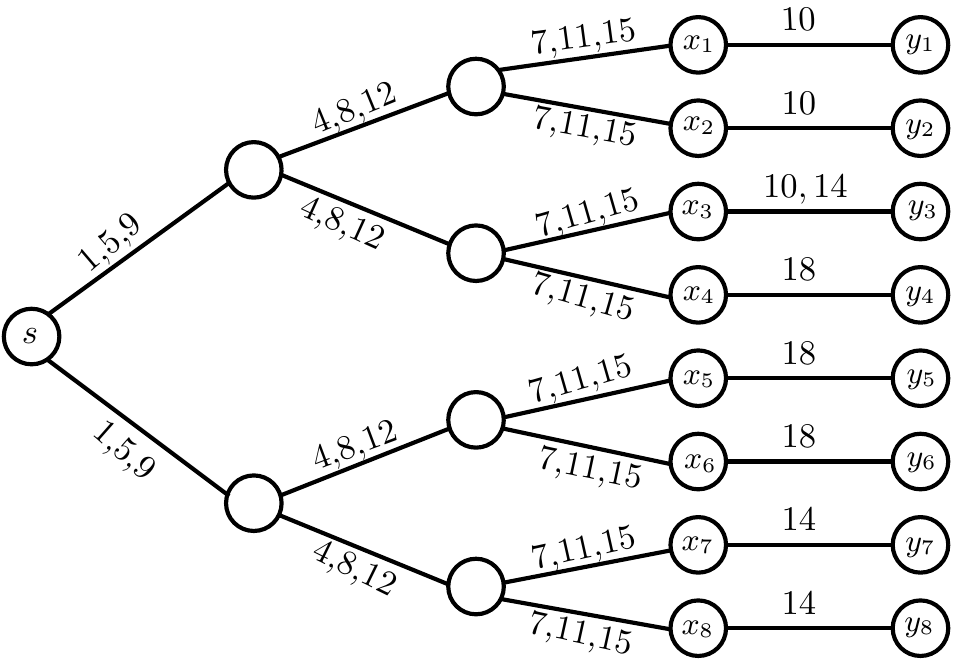}
    \caption{An example of the construction used in Theorem~\ref{thm:spread-uncon1}, when $\delta=3$. The subsets of the \setcover problem are $S_1=\{u_1,u_2,u_3\}$, $S_2=\{u_3,u_7,u_8\}$, $S_3=\{u_4,u_5,u_6\}$ and $\delta=3$.}
    \label{fig:objective1}
\end{figure}

The goal of this construction is to have the following properties:
\begin{itemize}
    \item Each transmission of vertex $s$ acts like a monotone wave, that goes from the vertex to the leaves, i.e., no vertex can influence its parent.
    \item For every $i\in [n]$ and $j\in [m]$, a transmission of vertex $s$ will influence a vertex $y_i$ only if $s$ transmits at a time step $t_j=(j-1)(\delta+1)+1$, such that set $S_j$ includes element $i$ in the \setcover instance.
\end{itemize}

\ifshort
In the next part, we formally list these properties which are vital for the proof of theorem \ref{thm:spread-uncon1}. 
\fi

\iflong
In the next part we formally prove these properties. Then we use them for the proof of Theorem \ref{thm:spread-uncon1}.
\fi

\begin{lemma}\label{lemma:flooding}
Consider constructed graph $P_1$. Let $t_j=(j-1)(\delta+1)+1$. If vertex $s$ transmits between time steps $t_j-(\delta-1)$ and $t_j$, then every vertex $v\in C$ at level $l$ will become active at time step $(l-2)\delta+(j-1)(\delta+1)+2$. Additionally the vertices become active once per such a transmission and an active vertex $v$ in level $l_1$ can never influence a vertex $w$ in level $l_2$, such that $l_2<l_1$ (the transmission goes from the root of the tree towards the leaves).
\end{lemma}

\iflong
\begin{proof}
First we will prove that no active vertex $v$ in level $l_2$ can ever influence a vertex $w$ in level $l_1$, such that $l_2>l_1$. This guarantees that the influence of a transmission cannot travel "backwards". Let us assume that the opposite is true i.e., there exists a vertex $v$ in level $l_2$ that can influence vertex $w$ in level $l_1$ from some transmission $\tau_1$ of vertex $s$ at time step $t_{j_1}=(j_1-1)(\delta+1)+1$. First, note that any such transmission $\tau_1$ from vertex $s$ must first reach vertex $v$ from its parent vertex $v^{\prime}$ via edge $vv^{\prime}$ at time step $[(l_1-1)-1]\delta+t_{j_1}$ and vertex $v$ becomes active at time step $[(l_1-1)-1]\delta+t_{j_1}+1$. Note now, that the next time step that edge $vv^{\prime}$ becomes available is $[(l_1-1)-1]\delta+t_{j_1+1}$. Therefore, the time difference between $v$ becoming active and the availability of edge $vv^{\prime}$ is $[((l_1-1)-1)\delta+t_{j_1+1}]-[((l_1-1)-1)\delta+t_{j_1}]=\delta$. This means that $v$ is active for $\delta$ time steps and in the next time step where edge $vv^{\prime}$ is available, $v$ is inactive. The same argument can be used to show that the children of $v$ cannot influence $v$ via transmission $\tau_1$. Thus, we have proven that a transmission cannot travel "backwards" and this also implies that every vertex can become active once per transmission (the graph is a tree and has no loops).
    
We will now prove the first part of our claim by using induction over the levels that each vertex belongs to. The base step is trivially shown by noticing that if vertex $s$ starts transmitting between time steps $t_{j_1}-(\delta-1)$ and $t_{j_1}$, then it will reach the vertices on $l=2$ using its adjacent edges that have label $t_{j_1}$. Now assume that vertex $s$ reaches every vertex $v$ until level $l_v$ via transmission $\tau_1$ at time step $t_{j_1}$. Since we showed that transmission cannot travel "backwards", a vertex $v$ at level $l_v$ can be influenced by transmission $\tau_1$ only via its parent $v^{\prime}$ via edge $vv^{\prime}$ which is available at time step $[(l_{v}-1)-1]\delta+t_{j_1}$ (by construction). Therefore, vertex $v$ becomes active at time step $[(l_{v}-1)-1]\delta+t_{j_1}+1$ and the next available edges with its children at level $l_v+1$ are available at time step $(l_{v}-1)\delta+t_{j_1}$. The time difference between the availability of the edges with the children and the time step where $v$ becomes active is $[(l_{v}-1]\delta+t_{j_1}]-[((l_{v}-1)-1]\delta+t_{j_1}+1]=\delta-1$. Thus, $v$ will still be active once the edges with its children become available and $v$ will influence them. This completes the proof. 
\end{proof}
\fi

% \begin{proposition}\label{obs:timing}
% Consider constructed graph $P_1$. For every $j\in [m]$, if vertex $s$ transmits at time step $t$, such that $t_j-(\delta-1)\leq t\leq t_j$, then no other vertex will become active by this transmission. 
% \end{proposition}

% \iflong
% \begin{proof}
% This trivially follows, by construction, from the fact that vertex $s$ will have no adjacent edges available while transmitting.
% \end{proof}
% \fi

\begin{lemma}\label{lemma:sets}

Consider constructed graph $P_1$. For every $j\in [m]$ and $i\in[n]$, a vertex $y_i$ will become active by the transmission of vertex $s$, only if the transmission starts between time steps $t_j-(\delta-1)$ and $t_j$, and only if $i\in S_j$.
\end{lemma}

\iflong
\begin{proof}
From Lemma \ref{lemma:flooding}, we know that, for every $j\in [m]$ a transmission between time steps $t_j-(\delta-1)$ and $t_j$ will influence vertex $x_i\in C$ and make $x_i$ active at time step $(h-1)\delta+t_j+1$. By construction, edge $x_iy_i$ is available at time step $h\delta+(t_j+1$ if and only if $i\in S_j$. Assume that $i\in S_j$. The time difference between $x_i$ becoming active and $x_iy_i$ being available is $(h\delta+t_j)-((h-1)\delta+t_j)=\delta$. Thus, $x_i$ will still be active once $x_iy_i$ becomes available and $y_i$ becomes active. Otherwise, if $i\notin S_j$, $x_i$ will become inactive before edge $x_iy_i$ become available.
\end{proof}
\fi

\begin{theorem}\label{thm:spread-uncon1}
For any $\delta \geq 1$, \problemone is \NP-hard and \Wtwo-hard when~parameterized~by~$b$ on tree graphs with degree $3$.
\end{theorem}

\begin{proof}

We claim that there exists a solution to \problemone on constructed graph $P_1$, such that $s$ influences $3n$ vertices, if $s$ transmits at $|b|$ time steps at most in the constructed graph, if and only if there is a solution $T = T_1 \cup T_2 \cup \cdots \cup T_b$ to the original \setcover instance.

Let $t_j=(j-1)(\delta+1)+1$. Assume that we have solution $T$ to the \setcover set problem. We can construct a solution for \problemone by having vertex $s$ transmit at each time step $t_j$, such that $S_j\in T$. We can guarantee that this is a solution for \problemone since every $i\in [n]$ is included in at least one set $S_j$ and by Lemma \ref{lemma:sets} every vertex $y_i$ will become active by at least one transmission of vertex~$s$.

For the reverse direction, consider that we have a solution to \problemone which is a transmission schedule $\tcal=(\tau_1, \tau_2,\dots, \tau_b)$. This means that if vertex $s$ transmits at time steps $\tau_1, \tau_2,\dots, \tau_b$, then every vertex in the graph will become active including vertices $y_1,y_2,\ldots, y_n$. To construct a solution for the \setcover problem we do the following. For every $t\in[b]$ and every $j\in [m]$, if $t_j-(\delta-1)\leq \tau_t\leq t_j$, we add set $S_j$ to solution $T$ of the \setcover problem (ignoring duplicate additions). By definition, the size of $T$ is at most $b$. Note also that every element $i\in [n]$ is included in $T$ since every vertex $y_i$ becomes active by at least one transmission of vertex $s$ and by Lemma \ref{lemma:sets} this only happens when $i\in S_j$. This completes the proof.
\end{proof}

\subsection{\problemtwo and \problemthree}

We will now show that \problemtwo is \NP-complete and \Wtwo-hard when~parameterized~by~$b$ on tree graphs with degree $3$. Containment in NP is straightforward, since given a transmission schedule, we can simulate the process and check the maximum active vertices for any time step. 
We will use a similar construction to the previous theorem and prove some key properties that we need; the construction differentiates in several parts in order to accommodate the different objectives. The proof will be again via a reduction from \setcover. Note also that the following construction/proofs can also be easily modified to show that \problemthree is \NP-hard and \Wtwo-hard and as such, we will not provide one.

\medskip
\noindent {\bf Construction.} 
We reduce from \setcover.
We construct a perfect binary tree $P_2$ with $n$ leaves, where the root of the binary tree is called $s$ and the leaves of the binary tree are called $x_1,x_2,\ldots,x_n$ from ``left to right" and let $h=\log 2n$ be the height of the current tree. Note that we will not update the value of $h$ once more vertices are added to the tree.  Also, note that such a tree always has $2n-1$ vertices and is uniquely defined. Let set $C$ contain every vertex of the current graph. For every $j\in [m+1]$ and for every $l\in [h]$, we add label $(l-1)\delta+h(j-1)(\delta+1)+1$ to every edge of the tree between levels $l$ and $l+1$. Then we construct vertices $y_1,y_2,\ldots, y_n$ and vertices $z_1,z_2,\ldots, z_n$ and for every $i\in [n]$, we add edge $x_iy_i$ and $y_i,z_i$. For every $i\in [n]$ and $j\in [m+1]$, if $i\in S_j$, we add label $h\delta+h(j-1)(\delta+1)+1$ to edge $x_iy_i$.  For every $i\in [n]$ and $j\in [h(m+1)]$, we add labels $(j-1)(\delta+1)+1,(j-1)(\delta+1)+2$ to edge $y_iz_i$. This completes the construction; see Fig.~\ref{fig:objective2} for an example.
%with $\delta=3$ and $S=\{S_1, S_2,S_3\}$, where $S_1=\{u_1,u_2,u_3\}$, $S_2=\{u_3,u_7,u_8\}$, $S_3=\{u_4,u_5,u_6\}$.

\begin{figure}
    \centering
    \includegraphics[width=0.8\textwidth]{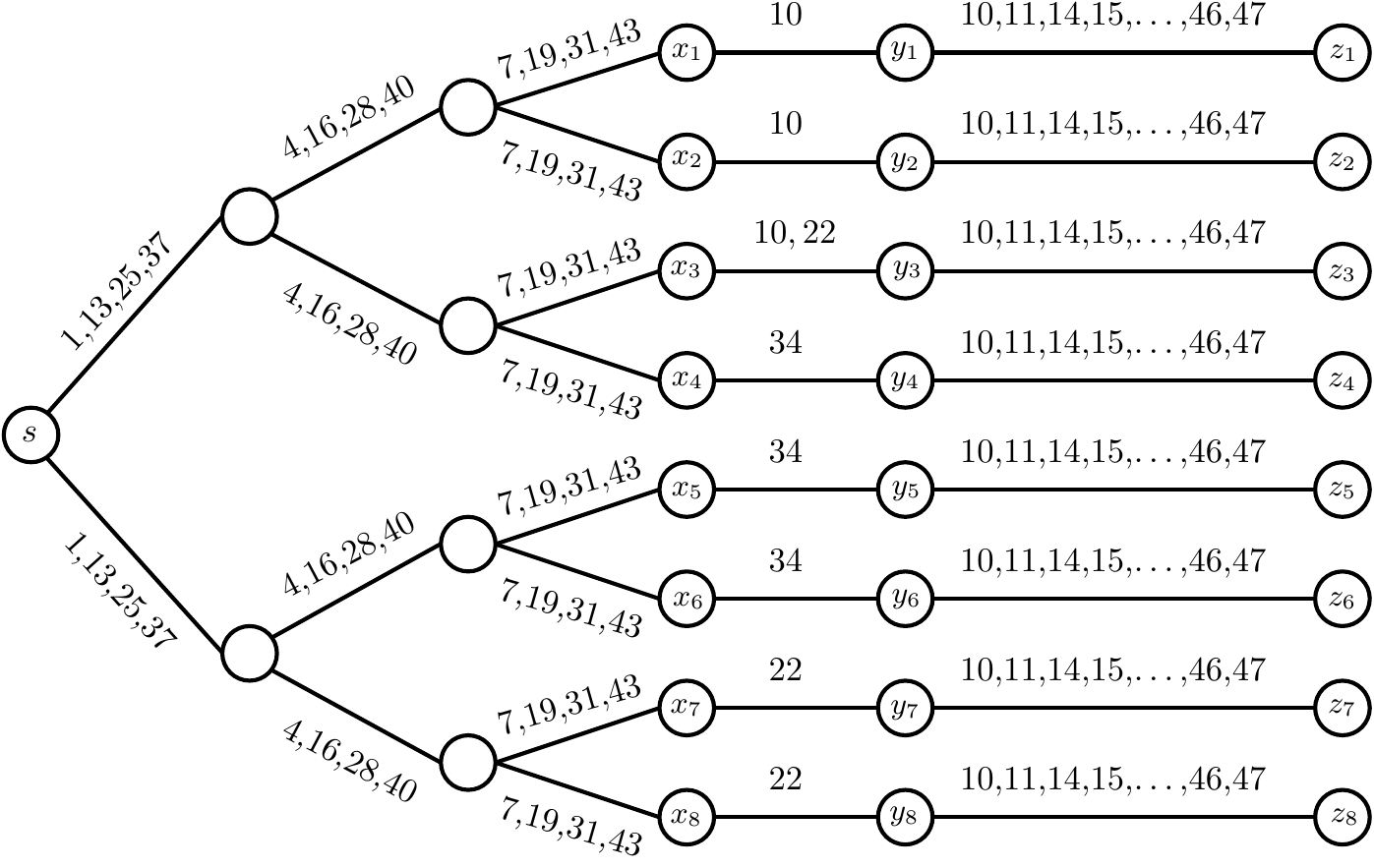}
    \caption{An example of the construction used to prove that \problemtwo is \Wtwo-hard, when $\delta=3$. The subsets of the \setcover problem are $S_1=\{u_1,u_2,u_3\}$, $S_2=\{u_3,u_7,u_8\}$, $S_3=\{u_4,u_5,u_6\}$.}
    \label{fig:objective2}
\end{figure}

The goal of this construction is to have the following properties:
\begin{itemize}
    \item Each transmission of vertex $s$ acts like a monotone wave, that goes from the vertex to the leaves, i.e., no vertex can influence its parent.
    \item For every $i\in [n]$ and $j\in [m]$, a transmission of vertex $s$ will influence a vertex $y_i$ only if $s$ transmits at a time step $t_j=h(j-1)(\delta+1)+1$, such that set $S_j$ includes element $i$ in the \setcover instance.
    \item For every $i\in [n]$ and $j\in [m+1]$, once vertex $y_i$ is influenced at some time step $h\delta+h(j_1-1)(\delta+1)+1$ by a transmission from vertex $s$, vertex $y_i,z_i$ will always be active at time step $h\delta+h(j-1)(\delta+1)+3$, where $j_1<j$.
\end{itemize}

\ifshort
Next, we formally list these properties which are vital for the proof of Theorem \ref{thm:spread-uncon2}. 
\fi

\iflong
In the next part we formally prove these properties. Then we use them for the proof of theorem \ref{thm:spread-uncon2}.
\fi

\begin{lemma}\label{lemma:flooding2}

Consider constructed graph $P_2$. Let $t_{j_1}=h(j-1)(\delta+1)+1$. If vertex $s$ transmits between time steps $t_{j_1}-(\delta-1)$ and $t_{j_1}$, then every vertex $v\in C$ at level $l$ will become active at time step $(l-2)\delta+t_{j_1}+2$. Additionally vertices $v\in C$ become active once per such a transmission and an active vertex $v$ in level $l_1$ can never influence a vertex $w$ in level $l_2$, such that $l_2<l_1$ (the transmission goes from the root of the tree towards the leaves).
\end{lemma}

\iflong
\begin{proof}
First we will prove that no active vertex $v$ in level $l_2$ can never reach a vertex $w$ in level $l_1$, such that $l_2>l_1$. This guarantees that the influence of a transmission cannot travel "backwards". Let us assume that the opposite is true i.e., there exists a vertex $v$ in level $l_2$ that can influence vertex $w$ in level $l_1$ from some transmission $\tau_1$ of vertex $s$ at time step $t_{j_1}$. First, note that any such transmission $\tau_1$ from vertex $s$ must first reach vertex $v$ from its parent vertex $v^{\prime}$ via edge $vv^{\prime}$ at time step $[(l_1-1)-1]\delta+t_{j_1}$, for some $j\in[m]$, and vertex $v$ becomes active at time step $[(l_1-1)-1]\delta+t_{j_1}+1$. Note now, that the next time step that edge $vv^{\prime}$ becomes available is $[(l_1-1)-1]\delta+t_{j_1+1}$. Therefore, The time difference between $v$ becoming active and the availability of edge $vv^{\prime}$ is $[((l_1-1)-1)\delta+t_{j_1}]-[((l_1-1)-1)\delta+t_{j_1+1}+1]=\delta$. This means that $v$ is active for $\delta$ time steps and in the next time step where edge $vv^{\prime}$ is available, $v$ is inactive. The same argument can be used to show that the children of $v$ cannot influence $v$ via transmission $\tau_1$. Thus, we have proven that a transmission cannot travel "backwards" and this also implies that every vertex can become active once per transmission (the graph is a tree and has no loops).
    
We will now show the first part of our claim by using induction over the levels that each vertex belongs to. The base step is trivially shown by noticing that for every $j\in[m+1]$, if vertex $s$ starts transmitting between time steps $t_j-(\delta-1)$ and $t_j$, then it will reach the vertices on $l=2$ using its adjacent edges that have label $t_j$. Now assume that vertex $s$ reaches every vertex $v$ until level $l_v$ via transmission $\tau_1$ at time step $t_{j_1}$. Since we showed that transmission cannot travel "backwards", a vertex $v$ at level $l_v$ can be influenced by transmission $\tau_1$ only via its parent $v^{\prime}$ via edge $vv\prime$ which is available at time step $[(l_{v}-1)-1]\delta+t_{j_1}$ (by construction). Therefore, vertex $v$ becomes active at time step $[(l_{v}-1)-1]\delta+t_{j_1}+1$ and the next available edges with its children at level $l_v+1$ are available at time step $(l_{v}-1)\delta+t_{j_1}$. The time difference between the availability of the edges with the children and the time step where $v$ becomes active is $[(l_{v}-1)\delta+t_{j_1}]-[((l_{v}-1)-1)\delta+t_{j_1}+1]=\delta-1$. Thus, $v$ will still be active once the edges with its children become available and $v$ will influence them. This completes the proof.  
\end{proof}
\fi

\begin{corollary}\label{corollary:levels}

Consider any pair of vertices $v,w\in C\setminus(s)$, where $v,w$ belong to different levels of $P_2$. There exists no such pair of vertices such that both vertices can be active at the same time step.  
\end{corollary}

\iflong
\begin{proof}
Consider any two vertices $v,w$ at levels $l_v$ and $l_w$ and a transmission from vertex $s$ at time step $t_{j_1}$. W.l.o.g. assume that $l_v<l_w$. By Lemma \ref{lemma:flooding2}, the minimum difference between the time steps where $v$ and $w$ become active by such a transmission is $\min_{v,w} |[(l_w-2)\delta+t_{j_1}+2]-[(l_v-2)\delta+t_{j_1}+2]|=\delta$. Therefore, assuming w.l.o.g. that $v$ becomes active before $w$, by the time step $w$ becomes active, at least $\delta$ time steps have passed since $v$ became active, and thus, $v$ is inactive.
\end{proof}
\fi

% \begin{proposition}\label{obs:timing2}

% Consider constructed graph $P_2$. For every $j\in [m+1]$, if vertex $s$ transmits at time step $t$, such that $t_j-(\delta-1)\leq t\leq t_j$, then no other vertex will become active by this transmission. 
% \end{proposition}

% \iflong
% \begin{proof}
% This trivially follows, by construction, from the fact that vertex $s$ will have no adjacent edges available while transmitting.
% \end{proof}
% \fi

\begin{lemma}\label{lemma:stop}
For every $i\in[n]$, vertex $y_i$ cannot influence vertex $x_i$.
\end{lemma}

\iflong
\begin{proof}
Since edge $x_iy_i$ is only available at time steps $h\delta+h(j-1)(\delta+1)+1$, for every $i\in[n]$ and $j\in[h(m+1)]$, we just have to show that vertex $y_i$ can never be active in those time steps. Because $\delta$ is smaller than the time difference between any two consecutive labels of edge $x_iy_i$, if vertex $x_i$ influences $y_i$, then $y_i$ cannot influence $x_i$ via the next time step that edge $x_iy_i$ is available (since $y_i$ would have become inactive by then), unless $z_i$ influences $y_i$ at some time step $h\delta+(j-1)(\delta+1)+2$. We will show that this is impossible by contradiction. Assume that vertex $z_i$ influences $y_i$ at some time step $h\delta+(j_y-1)(\delta+1)+2$ from some transmission of vertex $s$ that happened at time step $t_{j_1}$. Note now that, by construction, in order for $y_i$ to be active at time step $h\delta+(j_y-1)(\delta+1)+2$, vertex $x_i$ must have been active at some time step $h\delta+(j_x-1)(\delta+1)+1$, where $j_x<j_y$ and vertex $z_i$ has not influenced vertex $x_i$ yet. But we have already noticed that cannot that vertex $x_i$ cannot have been active at any time step $h\delta+(j_x-1)(\delta+1)+1$, unless $z_i$ has influenced vertex $x_i$ at least once.
\end{proof}
\fi

\begin{lemma}\label{lemma:sets2}

Consider constructed graph $P_2$. Let $t_j=h(j-1)(\delta+1)+1$. For every $j\in [m+1]$ and $i\in[n]$, a vertex $y_i$ will become active at time step $h\delta+t_{j_1}+3$ by the transmission of vertex $s$, only if the transmission starts between time steps $t_{j_1}-(\delta-1)$ and $t_{j_1}$, where $j_1<j$, and only if $i\in S_j$. Additionally, once vertex $y_i$ becomes active at time step $h\delta+t_{j_1}+3$ by the transmission of vertex $s$, vertices $y_i,z_i$ will be active at every time step $h\delta+t_{j_1}+3>h\delta+t_j+2$, for every $j\in [m+1]$.
\end{lemma}

\iflong
\begin{proof}
From Lemma \ref{lemma:flooding2}, we know that, for every $j\in [m+1]$, a transmission between time steps $t_j-(\delta-1)$ and $t_j$ will influence vertex $x_i\in C$ and make $x_i$ active at time step $(h-1)\delta+t_{j}+1$. Assume such a transmission at time step $t_{j_1}$. By construction, edge $x_iy_i$ is available at time step $h\delta+t_{j_1}$ if and only if $i\in S_j$. Assume that $i\in S_j$. The time difference between $x_i$ becoming active and $x_iy_i$ being available is $[h\delta+t_{j_1}]-[(h-1)\delta+t_{j_1}+1]=\delta-1$. Thus, $x_i$ will still be active once $x_iy_i$ becomes available and $y_i$ becomes active. Otherwise, if $i\notin S_j$, $x_i$ will become inactive before edge $x_iy_i$ becomes available. For the second part of the lemma, we are going to use induction. First, for the base step, note that vertex $z_i$ will become active at time step $h\delta+t_{j_1}+3$ since $y_i$ becomes active at time step $h\delta+t_{j_1}+2$ and edge $y_iz_i$ is available at time step $h\delta+t_{j_1}+2$ by construction. Assume now that, for some arbitrary $j_1^\prime>j_1$, both $y_i$ and $z_i$ are active at time step $h\delta+t_{j_1'}+3$. By the proof of Lemma \ref{lemma:stop}, we know that $y_i$ was not active at time step $h\delta+t_{j'}+1$, or else, $y_i$ would influence $x_i$. Therefore, at time step $h\delta+t_{j'}+1$, $z_i$ must have been active (or else both $y_i,z_i$ are inactive and the i-hypothesis does not hold). By construction, edge $y_iz_i$ is available at time step $h\delta+t_{j'}+1$ and $y_i$ becomes active at time step $h\delta+t_{j'}+2$. By construction, edge $y_iz_i$ is also available at time step $h\delta+t_{j'}+2$ and $z_i$ becomes active at time step $h\delta+t_{j'}+3$, and $y_i$ is still active. This completes the proof.
\end{proof}
\fi

\begin{theorem}\label{thm:spread-uncon2}
For any $\delta \geq 1$, \problemtwo is \NP-hard and \Wtwo-hard when~parameterized~by~$b$ on tree graphs with degree $3$.
\end{theorem}

\begin{proof}

We claim that there exists a solution to \problemtwo on constructed graph $P_2$, such that we maximize the number of vertices that are active at any one time step $t$ if $s$ transmits at $|b|$ time steps at most in the constructed graph, if and only if there is a solution $T = T_1 \cup T_2 \cup \cdots \cup T_b$ to the original \setcover instance.

First note, that due to Corollary \ref{corollary:levels}, the maximum vertices that can be influenced in the graph is exactly the set that contains vertices $x_i,y_i,z_i$. Let $t_j=h(j-1)(\delta+1)+1$. Assume that we have solution $T$ to the \setcover set problem. We split, the analysis into two cases: (i) $1\leq \delta \leq 2$ and (ii) $\delta \geq 2$. For case (i), we can construct a solution for \problemtwo by having vertex $s$ transmit at each time step $t_j$, such that $S_j\in T$. We can guarantee that this is a solution for \problemtwo since every $i\in [n]$ is included in at least one set $S_j$ and by Lemma \ref{lemma:sets2}, for every $i\in [n]$, every vertex $y_i$ will become active by at least one transmission of vertex $s$. Additionally, due to Lemma \ref{lemma:sets2}, every vertex $x_i,z_i$ will be active at every time step $h\delta+t_{j_1}+3$ for $t_{j+1}>t_j$. Finally, $s$ performs one final transmission at $(m+1-1)(\delta+1)+1$, so that all vertices $x_i,y_i,z_i$ are active at time step $h\delta+h(m+1-1)(\delta+1)+3$. For case (ii), assume that the last set is called $S_f$. We create the same solution as case (i) but we do not add the final transmission. This is because once vertex $s$ transmits at time step $(f-1)(\delta+1)+1$, at time step $h\delta+h(f-1)(\delta+1)+3$, every vertex $x_i$ will be active by the transmission at time step $(f-1)(\delta+1)+1$, or by a previous transmission. This was not true for case (i), because influence time is small, and vertices $x_i$ will have stopped being active at time step $h\delta+h(f-1)(\delta+1)+3$ by the transmission that fired at time step $(f-1)(\delta+1)+1$.

For the reverse direction, we again split the analysis in two cases: (i) $1\leq \delta \leq 2$ and (ii) $\delta \geq 2$. Consider that we have a solution to \problemtwo which is a transmission strategy $\tcal=(\tau_1, \tau_2,\dots, \tau_b)$. For case (i) this means that if vertex $s$ transmits at time steps $\tau_1, \tau_2,\dots, \tau_b$, then due to Lemma \ref{lemma:sets2}, vertices $x_i,y_i,z_i$ in the graph will be active at time step $h\delta+h(m+1-1)(\delta+1)+3$. To construct a solution for the \setcover problem we do the following. For every $t\in[b-1]$ and every $j\in [m]$, if $t_j-(\delta-1)\leq \tau_t\leq t_j$, we add set $S_j$ to solution $T$ of the \setcover problem (ignoring duplicate additions). We do not add the last set $b$ because this transmission was used only to influence vertices $x_i$. By definition the size of $T$ is at most $b$. Note also that every element $i\in [n]$ is included in $T$ since every vertex $y_i,z_i$ becomes active by at least one transmission of vertex $s$ and by Lemmas \ref{lemma:stop},\ref{lemma:sets2} this only happens when $i\in S_j$. For case (ii), we do the exactly the same but we also add set $b$ to the solution.
\end{proof}

\iflong
\subsection{\problemfour}
We prove that \problemfour is NP-complete and W[2]-hard when parameterized by $b$ on trees of maximum degree 4.
Containment in NP is straightforward, since given a transmission schedule, we can simulate the process and check whether the inactivation time of each node is at most $d$.
Our reduction is independent of the choice of $d$, as we can assure all vertices to be inactive for at most $d=1$ time step, and any efficient algorithm for $d>1$, it would also solve the constructed problem, thus providing a way to compute \textsc{SetCover} efficiently.
First, we describe the construction and prove some key properties our construction satisfies, before we provide the proof of the theorem.\\

\renewcommand{\P}{\ensuremath{P}}
\noindent \textbf{Construction.} We reduce from \textsc{SetCover}.
In what follows we will assume $n=2^k-1$ for some positive integer $k$.
Observe that this is without loss of generality, since we can augment any \textsc{SetCover} instance by adding a dummy set that contains the required number of elements.
    %and by asking for a solution of size $b+1$.
Refer to \Cref{fig:objective3} for an example of our construction. %, in particular the labels. %phrase

%--- Le tree
We construct a perfect binary tree $P$ with $n+1$ leaves, where the root of the tree is the source $s$ and its leaves are $x_1,x_2,\dots,x_n,x_a$ from ``left to right''. Let $h=\log 2n$ be the current height of the tree.
Note that such a tree always has $2n$ vertices and is uniquely defined.
%--- time labels in the tree -- TODO
%For every level $l\in [h]$ we add label $l$ and for every set from the \textsc{Set Cover} instance $j\in[m]$ we add labels $(l-1) + 2j(\delta + 1)$ to every edge of the tree between level $l$ and $l+1$.
For every set from the \textsc{SetCover} instance $j\in[m]$ and every level $l\in [h]$, we add labels $l$ and $(l-1) + 2j(\delta + 1)$ to every edge of the tree between level $l$ and $l + 1$. 
%--- Le appendix
Then, we construct vertices $y_1,y_2,\dots,y_n,y_a$, and for every $i\in [n]\cup \{a\}$ we add edge $x_iy_i$. %, as well as the edge $ay_a$.
For every $i\in[n]$ and $j\in[m]$, if $i\in S_j$ we add label $h+2j(\delta+1) $ to edge $x_iy_i$.
To edge $x_ay_a$ we add label $h+1$.
%--- second appendix
Next, we construct vertices $w_1,w_2,\dots,w_n,w_a$, and for every $i\in[n]\cup \{a\}$ we add edge $y_iw_i$. %, as well as the edge $y_aw_a$.
We add labels $h+2, h+3,\dots,h+2m(\delta+1)+1=t_{max}$ to every edge $y_iw_i$ with $i\in[n]\cup \{a\}$. %to edge $y_aw_a$ and to every edge $y_iw_i$ with $i\in[n], i=a$.
Note that the constructed tree now has $h+2$ levels. 
%--- additional excesses
Lastly, for every vertex $v$ on level $l$ of the original perfect binary tree \P, we construct two vertices $z_v$ and $z_v'$, and add edges $vz_v$ and $z_vz_v'$.
We add labels $l$ and $t_i=(l-1)+i(\delta+1)$ for all $i\in\mathbb{N}_{>0}$ with $t_i\leq t_{max}$ to the edge $vz_v$.
And to edge $z_vz_v'$ we add labels $l+1,l+2,\dots,t_{max}$.
This completes the construction.
% \begin{figure}
%     \centering
%     \includegraphics[scale=0.2]{fig_MNVT_example.jpg}
%     \caption{An example of the construction used to prove that \problemfour is $W[2]$-hard, when $\delta=3$, $d=1$. The subsets of the \textsc{SetCover} problem are $S_1=\{u_1,u_2,u_3\}$, $S_2=\{u_3,u_6,u_7\}$ and $S_3=\{u_4,u_5\}$.}
%     \label{fig:objective3}
% \end{figure}
\begin{figure}
    \centering
    \includegraphics[width=\columnwidth]{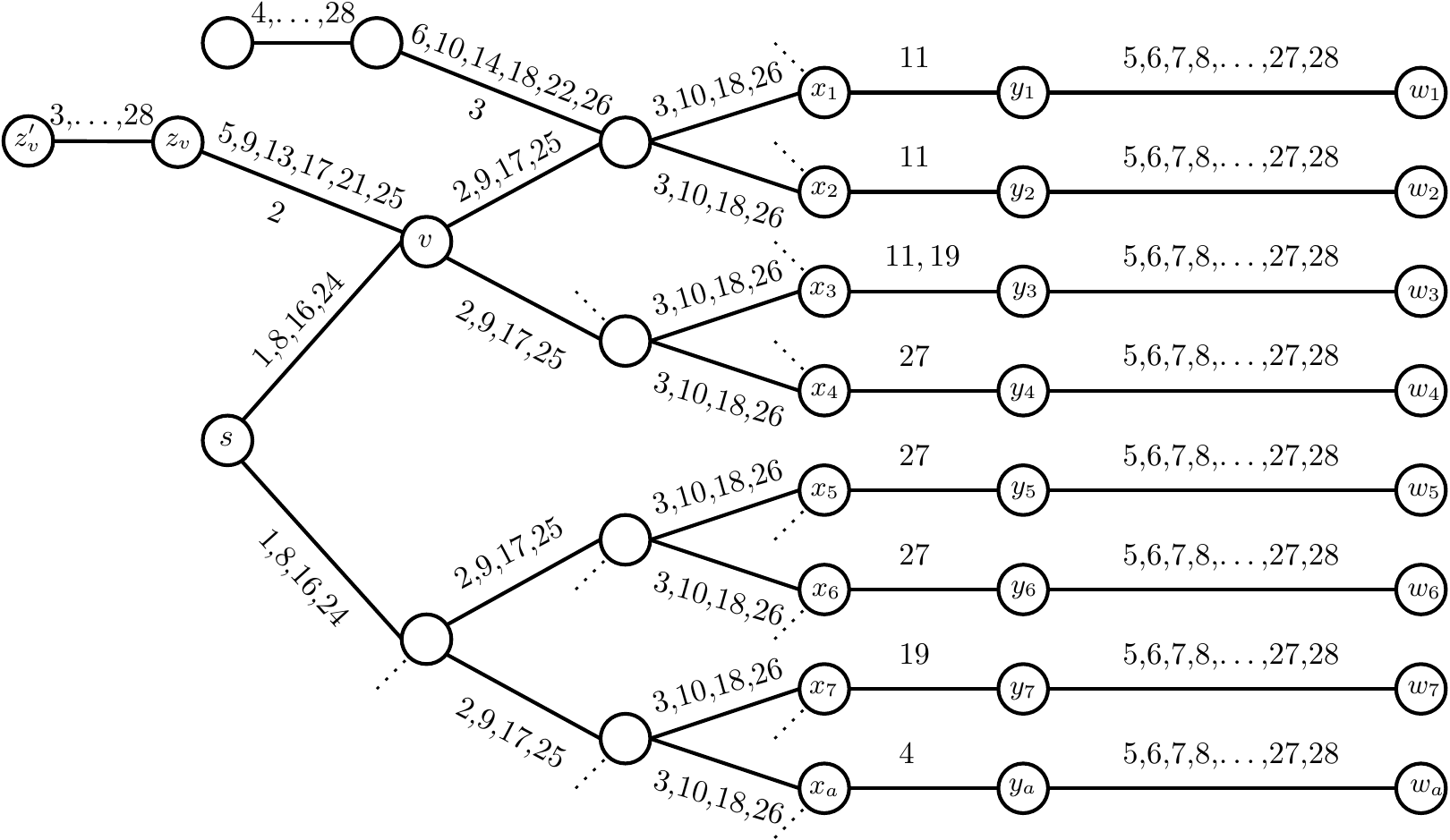}
    \caption{An example of the construction used to prove that \problemfour is \Wtwo-hard, when $\delta=3$, $d=1$. The subsets of the \textsc{SetCover} problem are $S_1=\{u_1,u_2,u_3\}$, $S_2=\{u_3,u_6,u_7\}$ and $S_3=\{u_4,u_5\}$.}
    \label{fig:objective3}
\end{figure}

%--- Intuition
The goal of this construction is to have the following properties:
\begin{itemize}
    \item Source $s$ has to transmit at time step 1 be able to influence $y_a$ and $w_a$, i.e., maximize the number of active vertices.
    \item Each transmission of $s$ acts like a monotone wave with one rebound that goes from the source to the leaves, i.\,e., no vertex can influence the parent of its parent.
    %Every but the first wave represents one set of the Set Cover instance.
    \item Once any vertex $v$ of the original binary tree $P$ is active for the first time, it is never inactive for more than one consecutive time step.
    In particular, $v$ enters a recurring sequence of activation: $v$ becomes active, stays active for $\delta$ time steps, is inactive after $\delta+1$, active again for the following $\delta$ time steps and after a total of $2(\delta+1)$ time steps it is active if and only if $s$ transmits at a time step $t_j=2j(\delta +1)$, for some set $S_j$ of the \textsc{SetCover} instance.
    \item For every $i\in[n]$ and $j\in[m]$, a transmission of $s$ will influence $y_i$ only if $s$ transmits at a time step $t_j=2j(\delta +1)$, such that set $S_j$ contains element $i$ in the \textsc{SetCover} instance.
    \item Once $y_i$ or $y_a$ is active for the first time, $w_i$, resp. $w_a$, will be active in the next time step. Both will stay active until $t_{max}$.
\end{itemize}
Next, we formally prove these properties vital for the proof of \Cref{thm:MNVT_NPhard}.
Recall, $t_{max}=h+2m(\delta+1)+1$ in our construction.

\begin{lemma}\label{lem:first_wave}
    Source $s$ has to transmit at time step 1 to maximize the overall number of activated vertices. In particular, there is no transmission schedule $T$ activating $y_a$ and $w_a$ where $1\notin T$.
\end{lemma}
\begin{proof}
%    The claim follows, as the only temporal path from $s$ to $y_a$ and $w_a$ starts at time step 1 and takes the edge $ay_a$, which has only the label $h$.
    The claim follows, as any temporal path from $s$ to $y_a$ and $w_a$ has to take the edge $x_ay_a$ which only has the label $h$. Therefore, there is only the temporal path starting at time step 1.
\end{proof}

\begin{lemma} \label{lem:wave}
    Consider the originally constructed binary tree \P.
    If source $s$ transmits at a time step $t=1$ or $t=2j(\delta+1)$ for some $j\in[m]$, every vertex $v\in V(\P)$ at level $l$ is active at time step $(l-1)+t$.
    In particular, every $v\in V(\P)$ is active for the first time at time step $l$.
\end{lemma}
\begin{proof}
Let $P$ be the originally constructed binary tree, i.e., the height of $P$ is $h$ and its leaves are $x_1,x_2,\dots, x_n,x_a$.
Assume source $s$ transmits at time step $t=1$. 
    There is a temporal path from $s$ to every leaf $x_i$ with $i\in [n]\cup \{a\}$, following the $l$ labels on every edge between level $l$ and $l+1$.
    Thus, every vertex at level $2$ is active at time step $2$ because of the transmission of $s$ at time step $1$.
    These vertices influence their children at level $3$, and so on.
    Overall, every vertex $v\in V(\P)$ at level $l$ is first active at time step $(l-1)+1$.
Now, assume $s$ transmits at time step $t=2j(\delta+1)$ for some $j\in[m]$.
    There is a temporal path from $s$ to every leaf $x_i$ with $i\in[n]\cup\{a\}$, following the $(l-1) + 2j(\delta + 1)$ labels on every edge between level $l$ and $l+1$.
    Thus, every vertex at level $2$ is active at time step $1+ 2j(\delta+1)$ because of the transmission of $s$ at time step $2j(\delta+1)$.
    These vertices influence their children at level $3$, and so on.
    Overall, every vertex $v\in V(\P)$ at level $l$ is active at time step $(l-1)+2j(\delta+1)$.
\end{proof}

\begin{lemma} \label{lem:inactive_time}
    Consider the originally constructed binary tree \P.
    After being active for the first time at time step $l$, every vertex $v\in V(\P)$ at level $l$ is always inactive at time step $t_i=(l-1)+i(\delta+1)$ for all uneven $i>0$ with $t_i<t_{max}$.
    For all even $i$ with $t_i<t_{max}$, vertex $v$ is inactive at $t_i$ if and only if $s$ does not transmit at time step $i(\delta+1)$.
    At~all other time steps vertex $v$ is active.
    This reactivation of $v$ is achieved by $z_v$ which is active at all time steps from $l+1$ until $t_{max}$.
\end{lemma}
\begin{proof}
    Let $v$ be a vertex at level $l$. 
    Refer to \Cref{fig:objective3_cutout} for an illustration of all edges adjacent to such a vertex together with their labels.
    By \Cref{lem:first_wave}, we know $v$ is first active at time step $l$.
    We prove the claim via induction over $i$.
    \begin{figure}
        \centering
        \includegraphics[width=0.75\columnwidth]{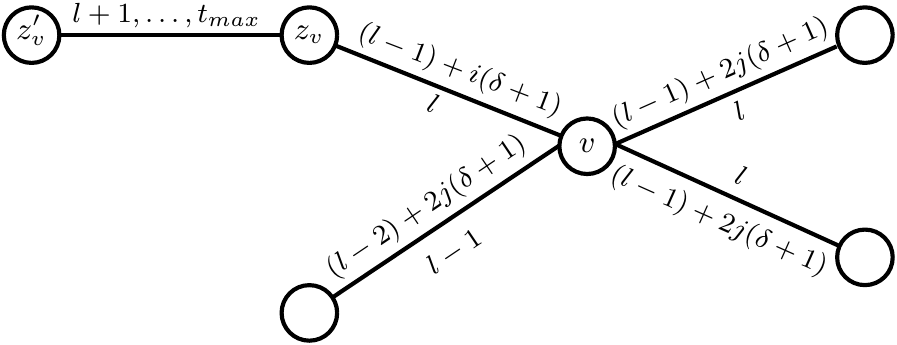}
        \caption{Cutout from the construction showing one of the vertices $v\in V(\P)$ at level $l$ of the originally constructed binary tree \P, together with its parent to the left and children to the right, as well as the additional leaves $z_v,z_v'$. On the edges, the corresponding labels are displayed.}
    \label{fig:objective3_cutout}
    \end{figure}
    
    Let $i=1$. As $v$ becomes active at time step $l$, it stays active until time step $l+\delta-1$ (included).
    Because of label $l$ on the edge $vz_v$, $z_v$ is active for the first time at time step $l+1$ and influences $z_v'$, which is active for the first time at time step $l+2$.
    As there are no adjacent edges of $v$ with a label in $[l+1,l+\delta-1]$, $v$ could only be reactivated by an edge with label $l$, i.e., either by one of its children in $P$ or by $z_v$.
    But those vertices are inactive until time step $l$, thus they cannot influence $v$ in time step $l$.
    Therefore, $v$ is inactive at time step $l+\delta$. %$ = (l-1)+i(\delta+1)$.
%    Now, we argue why $v$ is active at time step $l+(\delta+1)$.
    Because the edge $z_vz_v'$ has every label from $l+1$ to $t_{max}$, $z_v$ and $z_v'$ reactivate each other constantly starting at time step $l+1$ and stay active until $t_{max}$.
    Therefore, because of label $(l-1)+(\delta+1)$ on $vz_v$, $z_v$ will reactivate $v$ at time step $l+(\delta+1)$.

\newcommand{\actime}{\ensuremath{t_{act}}}
    Now, assume the claim holds for any vertex $v$ of $P$ at level $l$ for $0<j<i$.
    Meaning, $v$ is always inactive at $t_j=(l-1)+j(\delta+1)$ for uneven $0<j<i$ with $t_j<t_{max}$, and for even $j$ if and only if $s$ does not transmit at time step $j(\delta+1)$.
    In particular, $v$ is definitely active at $\actime=(l-1)+(i-1)(\delta+1)+1$, because $z_v$ is active and since label $(l-1)+(i-1)(\delta+1)$ is on $vz_v$.
    Therefore, $v$
    \begin{itemize}
        \item stays active until $\actime+\delta-1$, and
        \item will become inactive at time step $\actime+\delta=(l-1)+i(\delta+1)$,
        \item if not reactivated between $\actime$ and $\actime+\delta-1$.
    \end{itemize}
    The only edge adjacent to $v$ with a label in $[\actime,\actime+\delta-1]$ is the edge $p_vv$ to its parent $p_v$ with label $\actime+\delta-1$.
    By \Cref{lem:wave}, $p_v$ is active at time step $\actime+\delta-1=(l-2)+i(\delta+1)$ only if $i$ is even and source $s$ transmits at time step $i(\delta+1)$.
%    In th case $v$ is active at time step $\actime + \delta$.
    Otherwise, $p_v$ is inactive at time step $\actime+\delta-1$ and therefore cannot influence $v$.
%    Because of label $l+i(\delta+1)-1$ on $vz_v$, $z_v$ will reactivate $v$ at time step $l+i(\delta+1)$.
    Thus, $v$ is definitely inactive at time steps $(l-1)+i(\delta+1)$ for uneven $0<j<i$ with $t_j<t_{max}$, and for even $j$, $v$ is inactive if and only if $s$ does not transmit at time step $i(\delta+1)$. %At $(l-1)+i(\delta+1)$ all other time steps, $v$ is active.
\end{proof}

\begin{lemma} \label{lem:setcover_correlation}
    For every $i\in[n]$, any $y_i$ can only be active for the first time at a time step $h+2j(\delta+1)$ for some $j\in[m]$ with $i\in S_j$. This happens if and only if source $s$ transmits at time step $2j(\delta+1)$.
    Additionally, once a vertex $y_i$ with $i\in[n]\cup\{a\}$ is active, the vertices $y_i$, $w_i$ are active at every time step until $t_{max}$.
%    Similarly, once vertex $y_a$ is active, the vertices $y_a$, $w_a$ are active at every time step until $t_{max}$.
\end{lemma}
\begin{proof}
Let $i\in[n]$. The only temporal paths from source $s$ to $y_i$ take the edge $x_iy_i$ which has labels $h+2j(\delta+1)$ for all $j\in[m]$ with $i\in S_j$.
If $s$ transmits at time step $2j(\delta+1)$, $x_i$ is active at time step $h+2j(\delta+1)$ by \Cref{lem:wave}.
Now, only if $i\in S_j$, $x_i$~can influence $y_i$, which is then active at time step $h+2j(\delta+1)+1$.
If, on the other hand, $s$ does not transmit at time step $2j(\delta+1)$ for some $j$ with $i\in S_j$, by \Cref{lem:inactive_time}, $x_i$ is inactive at $h+2j(\delta+1)$ and cannot influence $y_i$. 

The second claim follows immediately as the edge between $y_i$ and $w_i$, for any $i\in[n]\cup\{a\}$, has all labels in $[h+2,t_{max}]$, and the earliest that any $y_i$ is active is $h+2$.
\end{proof}

\begin{theorem} \label{thm:MNVT_NPhard}
    For any $\delta\geq 1$ and $u\geq 1$, \problemfour is NP-hard and W[2]-hard when parameterized by the budget $b$ on tree graphs with degree 4.
\end{theorem}
\begin{proof}
We claim that there exists a solution to \problemfour on the constructed graph, such that we maximize the number of overall activated vertices if $s$ transmits at $\lvert b\rvert+1$ time steps in the constructed graph, if and only if there is a solution $T= T_1\cup T_2\cup \cdots\cup T_b$ to the original \textsc{SetCover} instance.

%-------------- SetCover --> MNVT
First note that, due to \Cref{lem:first_wave}, \Cref{lem:wave} and \Cref{lem:setcover_correlation}, the whole set of vertices can be influenced in the graph if there is a solution to the \textsc{SetCover} instance.
Assume a solution $T$ to the \textsc{SetCover} instance and let $t_j=2j(\delta +1)$.
We construct a solution for \problemfour by having source $s$ transmit at time step 1 and each time step $t_j$ such that $S_j\in T$.

We can guarantee that this is a solution for \problemfour, since every $i\in[n]$ is included in at least one set $S_j$, and by \Cref{lem:setcover_correlation}, for every $i\in[n]$, every vertex $y_i$ and with it $w_i$ will be influenced by at least one transmission of source $s$.
The transmission at time step 1 guarantees $s$ influencing $y_a$ and $w_a$, see \Cref{lem:first_wave}.
Lastly, all vertices $v$ of the binary tree $P$ and their artificial leaves $z_v$ are influenced by $s$, due to \Cref{lem:wave} and \Cref{lem:inactive_time}.
Additionally, by \Cref{lem:inactive_time,lem:setcover_correlation} we know that no vertex is inactive for more than one time step in a row.
Therefore, the transmission schedule $T$ is a solution to the \problemfour problem.

%-------------- MNVT --> SetCover
For the reverse direction, consider a solution to \problemfour which is a transmission schedule $\mathcal{T}=(1,\tau_1,\tau_2,\dots,\tau_{b})$.
The transmission schedule has to start with time step 1 due to \Cref{lem:first_wave} which implies together with \Cref{lem:wave} that all vertices on level $l$ of the binary tree $P$ are active for the first time at time step $l$ and, due to \Cref{lem:inactive_time}, are inactive for at most one consecutive time step.
Additionally, as source $s$ transmits at time steps $\tau_1,\tau_2,\dots,\tau_{b},$ due to \Cref{lem:setcover_correlation}, vertices $y_i,w_i$ in the graph are active at time step $h+2j(\delta+1)+1$, resp. $h+2j(\delta+1)+2$, for some $j$.
To construct a solution $T$ for the \textsc{SetCover} instance, we do the following.
For every $t\in[b]$ and every $j\in[m]$, if $\tau_t=2j(\delta+1)$, we add set $S_j$ to solution $T$.
By definition the size of $T$ is $b$. Also note that every element $i\in[n]$ is included in $T$ since every vertex $y_i$ is influenced by at least one transmission of vertex $s$, and by \Cref{lem:wave} and \Cref{lem:setcover_correlation} this happens if and only if $i\in S_j$.
Therefore, $T$ is a set cover and thus a solution to the \textsc{SetCover} instance.
\end{proof}
\fi

\subsection{Approximation Algorithm}\label{subsec:ApproxAlgo}

In this subsection, we show that all of the problems we study admit a constant factor approximation. That is, while the problems are \NP-hard, we will show that for each of the problems there exists a polynomial time algorithm that either finds a transmission schedule $\tcal=(\tau_1, \tau_2,\dots, \tau_b)$ that achieves at least $ 1-{\frac {1}{e}}\approx 0.632$ fraction of the influence target $k$, or we correctly output that the influence target $k$ cannot be achieved. 
\iflong 
In order to show this, we prove that our studied problems are not only harder that \setcover problem  -- or its maximization version \textsc{MaximumCoverage} which is formally defined later -- as we showed in the previous subsection, but they are actually equivalent to this problem. 
\fi
More precisely, we give a reduction from each of our problems to \textsc{MaximumCoverage} \ifshort (the maximization version of \setcover) \fi that preserves both influence target $k$ and budget $b$; here $b$ will be the number of sets we wish to select and $k$ the number of elements of the universe we wish to cover. Moreover, we show that there is one-to-one correspondence between transmission schedules and the collections of selected sets.

The following lemma serves as the main tool in this section\iflong and consecutive results are mostly relatively straightforward applications of this result\fi. It basically states that if we have some transmission schedule $\tcal$, then the set of all vertices active at some time step $t$ is precisely the union over all $\tau\in \tcal$ of the vertices active at time step $t$ if we only transmit at time step $\tau$. 
\iflong 
The proof is relatively standard induction on $t$, however, we need to prove something slightly stronger, that is the counter of a vertex at time step $t$ when we follow the transmission schedule $\tcal$ is exactly the maximum counter value at that vertex if we transmit separately at each of possible time steps in $\tcal$.

\fi

\begin{lemma}\label{lem:ApproxMain}

Given a temporal graph $\gcal := \tuple{G, \ecal}$ with lifetime $t_{\max}$, a source vertex $s$, a time step $t\in [t_{\max}]$, integer $\delta > 0$, a transmission schedule $\tcal$, it holds that $(\delta, \tcal)$-$\activeVertex_t(\gcal,s) = \bigcup_{\tau\in \tcal} (\delta, \tau)$-$\activeVertex_t(\gcal,s)$. 
\end{lemma}

\iflong
\begin{proof}
Let us begin by defining some auxiliary sets.
Given a temporal graph $\gcal$ with lifetime $t_{\max}$, a source vertex $s$, a time step $\tau \in [t_{\max}]$, and integer $\delta > 0$, we define $(\delta, \tau)$-$\activeCounter_t(\gcal,s,v)$ to be the value of the counter of the vertex $v$ if $s$ transmits at time step $\tau$. Similarly, for a transmission schedule $\tcal$, we define $(\delta, \tcal)$-$\activeCounter_t(\gcal,s,v)$ to be the value of the counter of the vertex $v$ if $s$ transmits at each of the time steps in $\tcal$.

Now we are ready to prove the lemma. In fact, we will actually show something stronger: for every vertex $v\in V(G)$ it holds that $(\delta, \tcal)$-$\activeCounter_t(\gcal,s,v)= \max_{\tau\in \tcal}(\delta, \tau)$-$\activeCounter_t(\gcal,s,v)$. 
Because the vertex is active if and only if its counter is non-zero, the moreover part follows straightforwardly by applying the first part of the lemma on all vertices in $V(G)$. We prove that 
$(\delta, \tcal)$-$\activeCounter_t(\gcal,s,v)= \max_{\tau\in \tcal}(\delta, \tau)$-$\activeCounter_t(\gcal,s,v)$ for all $v\in V(G)$ (including $s$) by induction on $t$. 
The statement is clearly true at the time step $1$. Since $1$ is the first step we can transmit, $(\delta, \tcal)$-$\activeCounter_1(\gcal,s,v) = (\delta, \tau)$-$\activeCounter_1(\gcal,s,v)= 0$ for all $v\in V(G)\setminus \{s\}$ and $(\delta, \tcal)$-$\activeCounter_1(\gcal,s,s) = \delta$ if $1\in \tcal$ and $(\delta, \tcal)$-$\activeCounter_1(\gcal,s,s) = 0$ otherwise. Now let us assume that for all vertices $v\in V(G)$ and some time step $t$ it holds that $(\delta, \tcal)$-$\activeCounter_t(\gcal,s,v)= \max_{\tau\in \tcal}(\delta, \tau)$-$\activeCounter_t(\gcal,s,v)$. Let us consider a time step $t+1$. By the definition of the spreading process,  $(\delta, \tcal)$-$\activeCounter_{t+1}(\gcal,s,v) = \delta$ if and only if either $v=s$ and $t+1\in \tcal$ or $v\neq s$ and there exists a vertex $u\in V(G)$ such that $uv\in E_t$ and $(\delta, \tcal)$-$\activeCounter_{t}(\gcal,s,u) > 0$. However, by the inductive hypothesis $(\delta, \tcal)$-$\activeCounter_{t}(\gcal,s,u) = \max_{\tau\in \tcal}(\delta, \tau)$-$\activeCounter_{t}(\gcal,s,u)$ and there exists $\tau\in \tcal$ such that $(\delta, \tau)$-$\activeCounter_{t}(\gcal,s,u)>0$ and consecutively $(\delta, \tau)$-$\activeCounter_{t+1}(\gcal,s,v) = \delta$. Now, if $v=s$ and $t+1\notin \tcal$, then  $(\delta, \tcal)$-$\activeCounter_{t+1}(\gcal,s,s) = (\delta, \tcal)$-$\activeCounter_t(\gcal,s,s) - 1$, but the counter of $s$ also decreases in every individual transmission in $\tcal$, so $(\delta, \tcal)$-$\activeCounter_{t+1}(\gcal,s,s)= \max_{\tau\in \tcal}(\delta, \tau)$-$\activeCounter_{t+1}(\gcal,s,s)$. Finally, if for all $u\in V(G)$ such that $uv\in E_t$ we have that  $(\delta, \tcal)$-$\activeCounter_t(\gcal,s,u)=0$, i.e., $u$ is inactive, then the counter on $v$ decreases (or stays $0$), but by the induction hypothesis, $(\delta, \tcal)$-$\activeCounter_t(\gcal,s,u)=0\max_{\tau\in\tcal}(\delta, \tau)$-$\activeCounter_t(\gcal,s,u)=0$ and the counter on $v$ decreases by one (or stays $0$) also for each individual transmission separately as well. Therefore, the statement holds also for $t+1$, and by induction hypothesis it holds for all time steps in $[t_{\max}]$.
\end{proof}
\fi

Given the above lemma, it is rather straightforward to construct an instance of \textsc{SetCover} or \textsc{MaximumCoverage} from each of our problems. Before we do so, let us formally define the \textsc{MaximumCoverage} problem. An instance of \textsc{MaximumCoverage} consists of a collection $\mathcal{S}=\{S_1, S_2, \ldots, S_m\}$ of subsets over some universe $U$ and a positive integers $b$ and $k$, and we need to decide if there is a sub-collection $\tcal \subseteq \mathcal{S}$ of size $b$, such that $|\bigcup_{S\in \tcal} S| \ge k$. 
\iflong
To obtain our algorithms, we use the following well-known result that can be found for example in Chapter 3 of \cite{hochba1997approximation}.

\begin{theorem}\label{thm:MaxCoverageGreedy}
Given an instance $\tuple{\mathcal{S}, b, k}$ of \textsc{MaximumCoverage}, there is a polynomial time algorithm that either outputs a a collection $\tcal\subseteq \mathcal{S}$ such that $|\tcal| = b$ and $|\bigcup_{S\in \tcal}S|\ge (1-{\frac {1}{e}})k$ or correctly output that no collection of at most $b$ sets covers at least $k$ elmements. 
\end{theorem}

Given Lemma~\ref{lem:ApproxMain}, we can easily obtain the following two corollaries that establish the connection between our problems and \textsc{MaximumCoverage}.
\begin{corollary}\label{cor:MaxSpreadApproxRed}
Given a temporal graph $\gcal := \tuple{G, \ecal}$ with lifetime $t_{\max}$, a source vertex $s$, and integer $\delta > 0$, one can in polynomial time construct a collection $\mathcal{S} = \{S_1, S_2,\ldots, S_{t_{\max}}\}$ over $U=V(G)\setminus\{s\}$ such that for every $\tcal\subseteq [t_{\max}]$ it holds that $|\bigcup_{t\in[\tmax]}(\delta,\tcal)-\activeVertex_t(\gcal,s)| = |\bigcup_{i\in \tcal}S_i|$.
\end{corollary}
\begin{proof}
We can simply let $S_i = \bigcup_{t\in[\tmax]}(\delta,i)-\activeVertex_t(\gcal,s)$. By Lemma~\ref{lem:ApproxMain}, we get that $\bigcup_{t\in[\tmax]}(\delta,\tcal)-\activeVertex_t(\gcal,s) = \bigcup_{t\in[\tmax]}\bigcup_{i\in\tcal}(\delta,i)-\activeVertex_t(\gcal,s)$, which is clearly the same as $\bigcup_{i\in\tcal}\bigcup_{t\in[\tmax]}(\delta,i)-\activeVertex_t(\gcal,s) = \bigcup_{i\in\tcal}S_i$.
\end{proof}

\begin{corollary}\label{cor:MaxSViralTimeStepApproxRed}
Given a temporal graph $\gcal := \tuple{G, \ecal}$ with lifetime $t_{\max}$, a source vertex $s$, a time step $t\in [t_{\max}]$, and integer $\delta > 0$, one can in polynomial time construct a collection $\mathcal{S} = \{S_1, S_2,\ldots, S_{t_{\max}}\}$ over $U=V(G)\setminus\{s\}$ such that for every $\tcal\subseteq [t_{\max}]$ it holds that $|(\delta,\tcal)-\activeVertex_t(\gcal,s)| = |\bigcup_{i\in \tcal}S_i|$.
\end{corollary}
\begin{proof}
We can simply let $S_i = (\delta,i)-\activeVertex_t(\gcal,s)$ and the statement follows directly by Lemma~\ref{lem:ApproxMain}.
\end{proof}

Corollary~\ref{cor:MaxSpreadApproxRed} together with Theorem~\ref{thm:MaxCoverageGreedy} immediately implies the following result.
\begin{theorem}
Given an instance $\tuple{\gcal,s,\delta,b,k}$ of \problemone, there is a polynomial time algorithm that either outputs a transmission schedule $\tcal$, such that $|\tcal|=b$ and $|\bigcup_{t\in[\tmax]}(\delta,\tcal)-\activeVertex_t(\gcal,s)|\ge (1-{\frac {1}{e}})k$ or correctly output that no transmission schedule of size at most $b$ leads to at least $k$ active vertices overall.
\end{theorem}

Combining Corollary~\ref{cor:MaxSViralTimeStepApproxRed} together with Theorem~\ref{thm:MaxCoverageGreedy} we get the following.
\begin{theorem}
Given an instance $\tuple{\gcal,s,\delta,b,k,t}$ of \problemthree, there is a polynomial time algorithm that either outputs a transmission schedule $\tcal$, such that $|\tcal|=b$ and $|(\delta,\tcal)-\activeVertex_t(\gcal,s)|\ge (1-{\frac {1}{e}})k$ or correctly output that no transmission schedule of size at most $b$ leads to at least $k$ active vertices at time step $t$.
\end{theorem}
Finally, repeating the algorithm of the above theorem $t_{\max}$ times we get.
\begin{theorem}
Given an instance $\tuple{\gcal,s,\delta,b,k}$ of \problemtwo, there is a polynomial time algorithm that either outputs a transmission schedule $\tcal$, such that $|\tcal|=b$ and $|\max_{t\in[\tmax]}(\delta,\tcal)-\activeVertex_t(\gcal,s)|\ge (1-{\frac {1}{e}})k$ or correctly output that no transmission schedule of size at most $b$ leads to at least $k$ active vertices at any time step.
\end{theorem}
\fi
\ifshort
To obtain our results we use the well-known result that for \textsc{MaximumCoverage}, there is a polynomial-time algorithm that either outputs a collection $\tcal\subseteq \mathcal{S}$ such that $|\tcal| = b$ and $|\bigcup_{S\in \tcal}S|\ge (1-{\frac {1}{e}})\cdot k$ or correctly outputs that no collection of at most $b$ sets covers at least $k$ elements. 

\begin{theorem}
There is a polynomial-time algorithm that finds a $(1-{\frac {1}{e}})$-approximate solution for \problemone, \problemtwo, and \problemthree.
\end{theorem}

\fi 

%%%%%%%%%%%%%%%%%%%%%%%%%%%%%%%%%%%%%%%%%%%%%%%%%%%%%%%%%%%%%%%%%%%%%%%
%%%%%%%%%%%%%%%%%%%%%%%%%%%%%%%%%%%%%%%%%%%%%%%%%%%%%%%%%%%%%%%%%%%%%%%%
\section{Window Constrained Schedules}
\label{sec:window}
In this section we consider two types of so called ``window constraint'' schedules, where we are only interested in transmission schedules satisfying some additional constraints. 
First we study fixed-window schedules. There the lifetime of the temporal graph is split into a number of disjoint time intervals and the transmission schedule needs to have exactly one transmission in each of the intervals. 
Then, we shift our attention to $(x,y)$-shifting window schedules, where the difference between two consecutive transitions should be between $x$ and $y$. 
Both scenarios are associated with a new natural parameter: the {\em size}, $w$,  of the window. 

Note, if we do not restrict size of the window, then the results from the previous section extend rather straightforwardly. To see this, consider the fixed window case. Assume that we have an instance of \problemone for the unconstrained setting, on the temporal graph $\gcal = \tuple{G,\ecal}$.
%$\tuple{\gcal = \tuple{G,\ecal}, s, b,\delta, k}$ 
Then, we could just create a new temporal graph $\gcal' = \tuple{G',\ecal'}$ such that $G'=G$ and $\ecal'$ is a concatenation of $b$ copies of the sequence $\ecal$.
\iflong In other words, $\gcal'$ is the temporal graph, where we repeat $\gcal$ $b$-many times. \fi
We then set the windows to have size $t_{\max}$ each, i.e., first window contains time steps $1$ to $\tmax$, second window from $\tmax$ to $2\tmax$ and so on. It is rather straightforward to see that this reduction immediately gives hardness for \problemone. 
It is also not too hard to verify that using similar reductions as in previous section would give hardness for \problemtwo and \problemthree, when the size of the window is not bounded by a constant. We can easily get hardness for shifting window schedules, using an argument that is nearly the same. This time, between the two consecutive copies of $\ecal$, we introduce $\tmax$ many time steps without any edge, and we let $x=\tmax$ and $y=2\tmax$. 

On the other hand, if both the budget, $b$, and window size, $w$, were parameters, then we would immediately get that the lifetime is bounded by $b$ times the (max) window size (or $y+1$ in the case of $(x,y)$-shifting windows), so an exhaustive search already gives an algorithm running in time $\binom{b\cdot w}{b}\cdot\poly(|V(G)|\cdot \tmax)$, where $w$ is the size of the window. 
For this reason, in the rest of the section we will consider the cases, where the budget is large (or unrestricted) and the size of the window is small. 
We will show that, unfortunately, the problem remains \NP-hard even for constant size windows. 
%%%%%%%%%%%%%%%%%%%%%%%%%%%%%%%%%%%%%%%%%%%%%%%%%%%%%%%%%%%%%%%%%%%%%%%%
\subsection{Fixed Window Schedules} 
In this section we prove that all three problems are \NP-hard even for very restricted settings.

\begin{theorem}\label{thm:fixedWindowProblemOne}

    \problemone with fixed window constraints is \NP-hard for every $\delta \ge 1$ even when window size is $2\delta$, in every time step there are at most $3$ active edges, every edge is active at most twice, and the underlying graph is a star with center the source $s$.     
\end{theorem}

\begin{proof}
    We show this by reduction from \textsc{VertexCover} on graphs with maximum degree three\iflong which is well known to be \NP-complete\fi~\cite{GareyJ79}. In the \textsc{VertexCover} problem, we are given a graph $H = \tuple{V,E}$ and an integer $\ell$, and the question is whether there exists a set of vertices $S$ such that $|S|\le \ell$ and for all $e\in E$ we have $|e\cap S|\ge 1$. 

    Now, let $\tuple{H,\ell}$ be an instance of \textsc{VertexCover} such that the degree of every vertex $h\in V(H)$ is at most three. We construct a temporal graph $\gcal =\tuple{G, \ecal}$ as follows. First, for the sake of presentation of the proof, let $|V(H)|=n$, $|E(H)|=m$, and let us order the vertices and edges of $H$ in an arbitrary but fixed order. That is let $V(H)= \{h_1, h_2,\ldots, h_n\}$ and $E(H)= \{e_1, e_2, \ldots, e_m\}$.

    The vertices of $\gcal$ are as follows: 
    \begin{itemize}
        \item the source vertex $s$;
        \item the set $V_E$ containing $m$ vertices, such that for every $e_j\in E(H)$ there is a vertex $v_j\in V_{E}$;
        \item the set $U$ containing $n$ vertices, such that for every $h_i\in V(H)$, there is a vertex $u_i\in U$.
    \end{itemize}
    The underlying graph $G$ is then a star with center the vertex $s$ and $\tmax = 2\delta\cdot(|V(H)| + 1)$. Every window will consist of $2\delta$ time steps and will be associated with a single vertex of $V(H)$. 
    \ifshort
    So, for every $i \in [n]$ the window from time step $2\delta\cdot(i-1)+1$ to time step $2\delta\cdot i$ is associated with vertex $h_i$.
    \fi
    \iflong
    So, window $1$-$2\delta$ is associated with $h_1$, window $2\delta+ 1$-$4\delta$ is associated with $h_2$,$\ldots$, window $2\delta(n-1)+ 1$-$2\delta\cdot n$ is associated with $h_n$. 
    \fi
    The last window is a ``dummy'' window that is only necessary in the case of $\delta=1$.
    In this case, we might want to transmit at the time step $2\delta\cdot|V(H)|$ to activate some vertices, which we can only do if there is one more time step to activate the vertices adjacent to $s$ at time step $2\delta|V(H)|$. Consider now the window associated with $h_i$. At the first step inside this window, i.e.  at $2\delta(i-1)+1$, there will be edges between $s$ and every $v_j\in V_E$ such that $h_i$ is incident with the edge $e_j$. Note that since the degree of $h_i$ is at most $3$, at most $3$ edges have label $2\delta(i-1)+1$. Moreover, in $(\delta+1)$-th time step inside the window ($2\delta(i-1)+\delta + 1$) there is an edge between $s$ and $u_i$. That is exactly one edge -- the edge $su_i$ -- has the label $2\delta(i-1)+\delta + 1$. 

Given the above construction of the temporal graph $\gcal$, we let $k = n+m-\ell$ and we claim that there exists a window constraint transmission schedule $\tcal$ such that $\left|\bigcup_{t\in[\tmax]}(\delta,\tcal)-\activeVertex_t(\gcal,s)\right| \geq k$ if and only if $H$ admits a vertex cover with at most $\ell$ vertices. To see this, we only need to show that it only makes sense to transmit in $1$st or $\delta+1$st step in each window. Given this, we then observe that transmitting in the $1$st time steps in the windows associated with the vertices of a vertex cover $S$ and in $\delta+1$st windows in the remaining time steps activates all but $|S|$ many vertices in $U$ (precisely the vertices associated with the vertices in $S$).
\iflong

    First let us assume that $S$ is a vertex cover of size at most $\ell$. Then we construct $\tcal$ as follows. 
    \begin{itemize}
        \item For every $h_i\in S$, we let $2\delta(i-1)+1\in \tcal$;
        \item For every $h_i\in V(H)\setminus S$, we let $2\delta(i-1)+\delta+1\in \tcal$;
        \item We let $2\delta\cdot|V(H)| + 1\in \tcal$.
    \end{itemize}
    Note that since $S$ is a vertex cover, then for $e_j\in E(H)$, there exists $h_i\in V(H)$ such that $h_i\in e_j$ (i.e., $h_i$ is incident with $e_j$). Hence, $v_j\in V_E$ gets influenced at the time step  $2\delta(i-1)+2\in \tcal$, because $2\delta(i-1)+1\in \tcal$ and the edge $sv_j$ has label $2\delta(i-1)+1$. It follows that all vertices in $V_E$ get activated at some time step. 
    On the other hand, for every $h_i\in V(H)\setminus S$ we transmit in the time step $2\delta(i-1)+\delta+1\in \tcal$ and the edge $su_i$ has the label $2\delta(i-1)+\delta+1$ so $u_i$ gets influenced in the consecutive step. Therefore, all $u_i\in U$ with $h_i\in V(H)\setminus S$ are activated. It follows that overall $|V_E| + |V(H)\setminus S| = m + n - |S|\ge m + n - \ell=k$ vertices get activated in total. 

    On the other hand, let $\tcal$ be a window constraint transmission schedule such that $\Bigl|\bigcup_{t\in[\tmax]}(\delta,\tcal)-\activeVertex_t(\gcal,s)\Bigr|\geq k$. First let us make few observations about $\tcal$ and $\gcal$. The difference between any two distinct time steps $\tau_1$, $\tau_2$ such that both $E_{\tau_1}$ and $E_{\tau_2}$ are nonempty is at least $\delta$ and hence one transmission can only affect one time step. We would like to assume that the transmission that affects time step in the window associated with $h_i$ also happened in that window. That is not always possible. If we transfer in time step $2\delta(i-1)+\delta+2$ (or any time later than that before $2\delta i$), the transmission is still active at time step  $2\delta i+1$. However, note that in this case we could not transmit between $2\delta(i-1)+2$ and $2\delta(i-1)+\delta+1$, as we already transmitted within that window. So $u_i$ never gets influenced. On the other hand, if we change our transmission schedule such that we transmit at $2\delta(i-1)+\delta+1$ and at $2\delta i+1$, then $u_i$ gets influenced, all edges that are influenced by transmission $2\delta(i-1)+\delta+2$ in the time step $2\delta i+1$ still get influenced, we only also removed original transmission of the window $2\delta i+1$-$2\delta (i+1)$. However, if we apply this argument always on the maximum $i$ such that we transmit in a time step after $2\delta(i-1)+\delta+2$ within the $i$-th window, then only possibly the vertex $u_{i+1}$ does not get influenced, but this is compensated by the fact that we now influence $u_i$ that we did not originally. It follows that this does not decrease the number of vertices we influence and we can assume that in every window we transmit between $1$-st and $\delta+1$-st step inside the window. Furthermore, notice that if for some $e_j = h_{i_1}h_{i_2}$ the vertex $v_j\in V_E$ is not influenced at all, then we can replace transmission $\tau_{i_1} = 2\delta(i_1-1)+r$ for some $r$ between $2$ and $\delta+1$ by  $\tau_{i_1}'= 2\delta(i_1-1)+1$ that remove $u_{i_1}$ from the influenced vertices, but adds $e_j$, so it does not decrease the number of influenced vertices by $\tcal$ overall. By repeating this argument, it follows that we can always assume that all vertices in $V_E$ are influenced at some point. Since, $\left|\bigcup_{t\in[\tmax]}(\delta,\tcal)-\activeVertex_t(\gcal,s)\right| \geq k = n+m-\ell$, it follows that at least $n-\ell$ vertices of $U$ are influence. Therefore there is a set $\scal\subseteq \tcal$ of are at most $\ell$ transmissions $\tau_i$ such that $\tau_i = 2\delta(i_1-1)+1$ for some $i\in [n]$. It is easy to verify that $S = \{h_i\mid 2\delta(i-1)+1\in \scal\}$ is a vertex cover of $G$ with at most $\ell$ vertices. 
\fi
\end{proof}

Similarly as in the previous section, we can rather straightforwardly modify the above proof to obtain the \NP-hardness for \problemtwo and \problemthree.

\begin{theorem}\label{thm:fixedWindowProblemTwoThree}

    \problemtwo and \problemthree with fixed window constraints are both \NP-hard even when the window size is $2\delta$, for every $\delta \ge 1$.
\end{theorem}
\iflong
\begin{proof}
    %Similarly as in the previous section, this follows rather straightforwardly from Theorem~\ref{thm:fixedWindowProblemOne}. 
    Let $\tuple{\gcal, s, b, \delta, k}$ be an instance of \problemone with window size $2\delta$, such that the underlying graph of $\gcal$ is a star with source $s$. Since, $s$ cannot get influenced by vertices in $V(G)\setminus \{s\}$, it suffices to make sure that once we influence a vertex in $V(G)\setminus \{s\}$ it will stay influenced forever. This can be done by adding a pendant vertex to every leaf of the star and make the edge between this new vertex and the original leaf active in every time step (have all the labels in $[\tmax]$. It is rather straightforward to see that once a neighbor $v$ of $s$ is influenced, then it will influence the pendant $p_v$ adjacent to it and these two vertices will be influencing each other forever.
    The only thing we need to pay attention to is that if $\delta=1$, then these two vertices alternate in which is influenced and else they are both influenced all the time. 
    Hence to get instance $\tuple{\gcal', s, b, \delta, k'}$ of \problemtwo we need to set $k'=k$ if $\delta=1$ and $k'=2k$ otherwise. Finally, for \problemthree we just ask for the number of active vertices to be maximized at the time step $\tmax$.
\end{proof}
\fi

\subsection{Shifting Window Schedules}
The hardness results for the shifting windows case follow rather easily from the proofs in the previous section. For \problemone, we follow the same reduction as in Theorem~\ref{thm:fixedWindowProblemOne}, but: we add additional $\delta$ time steps between any two consecutive windows; each time step will be empty, i.e. it will not contain any edges; we ask for $(2\delta, 4\delta)$-shifting window. This guarantees that we can always transmit in every window of size $3\delta$, either at time step $3i\delta + 1$ or at time step $3i\delta + \delta +1$, which are the only two time steps within the window than contain some edges. Moreover, the lower value of the window, $2\delta$, guarantees that $s$ cannot be active during both of the time steps within one window. 
\begin{theorem}\label{thm:shiftingWindowProblemOne}
    \problemone with $(2\delta, 4\delta)$-shifting window constraints is \NP-hard, for every $\delta \ge 1$, even when in every time step there are at most $3$ active edges, every edge is active at most twice, and the underlying graph is a star with the center in the source $s$.     
\end{theorem}

Following a similar argument to Theorem~\ref{thm:fixedWindowProblemTwoThree} \iflong of adding pendants to the leaves of the star in order to make vertices active forever (or alternating between original leaf and pendant in case of $\delta=1$) once activated\fi , we get\iflong the following\fi .  

\begin{theorem}\label{thm:shiftingWindowProblemTwoThree}
    \problemtwo and \problemthree $(2\delta, 4\delta)$-shifting window constraints is \NP-hard for every $\delta \ge 1$.
\end{theorem}

%%%%%%%%%%%%%%%%%%%%%%%%%%%%%%%%%%%%%%%%%%%%%%%%%%%%%%%%%%%%%%%%%%%%%%%%
%%%%%%%%%%%%%%%%%%%%%%%%%%%%%%%%%%%%%%%%%%%%%%%%%%%%%%%%%%%%%%%%%%%%%%%%
%%%%%%%%%%%%%%%%%%%%%%%%%%%%%%%%%%%%%%%%%%%%%%%%%%%%%%%%%%%%%%%%%%%%%%%%
\section{Schedules on Periodic Graphs}
\label{sec:periodic}
In this section we investigate the complexity of our three problems on periodic temporal graphs. 
A periodic temporal graph is given to us in the exactly same way as the temporal graph defined in Section~\ref{sec:prelims}. The only difference is that after reaching the time step $t_{\max}$ instead of 
stopping the time, the temporal graph repeats itself from time step $1$. More precisely, given a {\em temporal graph} $\gcal := \tuple{G, \ecal}$, where $\ecal = (E_1, E_2, \ldots, E_{\tmax})$, if we say that $\gcal$ is a periodic graph, then edge with label $i \in [\tmax]$ is available not only in time step $i$ but in every time step $i +j\cdot\tmax$, where $j\in \mathbb{N}$. In this case, we call $t_{\max}$ the \emph{period} of $\gcal$. 

Note that while in previous sections we only needed to simulate the spreading process until time step $t_{\max}$, this is now not the case and the spreading process continues infinitely. However, our questions still make sense and we can even upper bound the number of steps we need to simulate. 
It is not so difficult to see that there are at most $(\delta+1)^{|V(G)|}$ different combinations of counter values on the vertices of the graph $G$.
Thus, after simulating at most $(\delta+1)^{|V(G)|}\cdot t_{\max}$ steps of the spreading process, some combination of counters will appear at time steps $i +j\cdot\tmax$ and $i +j'\cdot\tmax$ for some $i \in [\tmax]$.
Since the process is deterministic, if there is no new transmission between these two time steps, we end up repeating exactly the same sequence of combinations of counter values (and hence active vertices), as we have between these two time steps.
Therefore, the time between any two consecutive transmissions in a solution to any of our problems, does not need to be more than $(\delta+1)^{|V(G)|}\cdot t_{\max}$ steps apart. 
Thus, for any such solution we need to simulate $(\delta+1)^{|V(G)|}\cdot t_{\max}\cdot (b+1)$ steps, where $b$ is the budget we are given. 

In what follows, we consider the situation when the period $t_{\max}$ is small, i.e., we study the parameterized complexity of the problem with respect to the parameter $t_{\max}$. 
Before we prove our results, we would like to highlight that restricting the period is necessary if we hope to get any positive results for the problem. 
This is because adding $\delta+1$ new labels $E_{t_{\max}+1}, E_{t_{\max}+2},\ldots, E_{t_{\max}+\delta+1}$ where $E_j=\emptyset$ for all $j$ between $t_{\max}+1$ and $t_{\max}+\delta+1$, gives us a periodic graph in which all vertices (other than $s$) are always inactive when the period starts.
Hence, the idea described above effectively reduces the problem on periodic temporal graphs to the problem when we are given non-periodic temporal graphs. 

Let us now consider \problemone. Since we only care about which vertices are activated at least once, but do not care about synchronicity, we immediately get from the application of Lemma~\ref{lem:ApproxMain} that we only need to consider transmitting within the first period.
So, if $b\ge t_{\max}$, then clearly an optimal solution is to just transmit $t_{\max}$ times in total -- in every time step within the first period.
On the other hand, if $b< t_{\max}$, then we can enumerate all  $\binom{t_{\max}}{b} < 2^{t_{\max}}$ possibilities to transmit in $b$ different time steps within the first period. However, there is still one issue to resolve: our upper bound $(\delta+1)^{|V(G)|}\cdot t_{\max}\cdot (b+1)$ on the number of steps we need to simulate is too large. In the next theorem we prove that, if we only care about whether a vertex has been activated or not, then we can overcome this large upper bound and show that the number of steps we need to simulate is actually only at most $(|V(G)|+1)\cdot t_{\max}\cdot \delta$.

\begin{theorem}\label{thm:PeriodicSpread1}

    \problemone can be solved in $2^{t_{\max}}\cdot \poly(|V(G)|)$ time on a periodic temporal graph $\gcal=\tuple{G,\ecal}$ with period $t_{\max}$.
\end{theorem}
\iflong
\begin{proof}
    \iflong Given two transmissions times $\tau_1 < \tau_2$ such that $\tau_2-\tau_1 = 0 \mod t_{\max}$, 
    it is straightforward to see that $(\delta,\tau_1)-\activeVertex_t(\gcal,s) = (\delta,\tau_2)-\activeVertex_{t+\tau_2-\tau_1}(\gcal,s)$. Therefore, $\bigcup_{t\in \mathbb{N}}(\delta,\tau_1)-\activeVertex_t(\gcal,s) = \bigcup_{t\in \mathbb{N}}(\delta,\tau_2)-\activeVertex_t(\gcal,s)$. Therefore, if we are looking for transmission schedule $\tcal$ that maximizes $\bigcup_{t\in \mathbb{N}}(\delta,\tcal)-\activeVertex_t(\gcal,s)$, it is a straightforward application of Lemma~\ref{lem:ApproxMain} (by replacing the order in which we do the union), we can always in $\tcal$ replace $\tau_1$ by $\tau_2$ and vice versa. Moreover, if $\tau_1$ is already in $\tcal$, adding $\tau_2$ does not change the set of activated vertices. Hence, we can assume that we only transmit within the first $t_{\max}$ time steps. Given this restriction, there are at most $2^{t_{\max}}$ possible transmission schedules and to determine whether we can activate at least $k$ vertices, it suffice to compute the set of activated vertices for each of the schedules of size $\min\{b,t_{\max}\}$. 
    It \fi 
    \ifshort As argued above the theorem, it \fi
    remains to show that given a fixed transmission schedule we can compute the set of activated vertices in polynomial time. To do so, we show that it suffices to simulate $(|V(G)|\cdot t_{\max}\delta$
    time steps after last time the source transmitted and after this step, no new vertex will be activated\iflong (there can still be new combinations of vertices activated at the same time, but they will consist of only the vertices that have already been active before)\fi . 
    
    From now on consider a fixed $\tcal\subseteq [t_{\max}]$ and a vertex $v$ that have been for the first time activated at time step $i$. We show that $i\le (|V(G)|+1)\cdot t_{\max}\cdot \delta$. The reason that $v$ got activated at time step $i$ is that there is a vertex $u$ that is active and a neighbour of $v$ at time step $i-1$ (i.e., $uv\in E_j$ for 
    $j=(i-2 \mod t_{\max})+1$). For that to be case, vertex $u$ was activated at time step $j \ge i-\delta$ either by its neighbour $w$ in time step $j-1$ or $u=s$ and it transmitted at $j$. If we continue this reasoning until we reach $s$ at some time step $i_0$, we can construct a sequence $(s, i_0), (u_1, i_1), (u_2, i_2), \ldots, (u_q, i_q) = (v,i)$ such that 
    for all $j\in [q]$ the vertex $u_j$ has been activated at the time $i_{j}$ because the vertex $u_{j-1}$ was its neighbor and at the same time active at the time step $i_{j}-1$.  
    Note that for all $j\in [q]$ we have $i_j - i_{j-1} \le \delta$ and $i_0\le t_{\max}$. Hence it suffices to show that $q\le |V(G)|\cdot t_{\max}$. Now assume, for the sake of contradiction, that $q > |V(G)|\cdot t_{\max}$, then there are two pairs $(u_j, i_j)$ and $(u_{j'}, i_{j'})$ such that $j < j'$, $u_j = u_{j'}$, and $i_j-i_{j'} = 0 \mod t_{\max}$. However, then the sequence $(s, i_0), (u_1, i_1), (u_2, i_2), \ldots, (u_j, i_j), (u_{j'+1}, i_{j'+1}- (i_j-i_{j'})), (u_{j'+2}, i_{j'+2}- (i_j-i_{j'})),\ldots,   (u_q, i_q - (i_j-i_{j'})) = (v,i-(i_j-i_{j'}))$ \ifshort contradicts the fact that $v$ has been activated for the first time in the time step $i$. \fi\iflong also give us a sequence such that for all $\ell\in [q]$ the vertex $u_\ell$ has been activated at the time $i_{\ell}$ because the vertex 
    $u_{\ell-1}$ was its neighbor and at the same time active at the time step $i_{\ell}-1$.  This is because the vertex $u_j=u_{j'}$ has been activated at the time step $i_j$ and because $i_{j'}- i_j$ is divisible by $t_{\max}$, it follows that edge $u_{j}u_{j'+1}$ is also active at time step $i_{j'+1}- (i_j-i_{j'}) - 1$, as it is active at time step $i_{j'+1}-1$. So, $u_{j'+1}$ gets indeed activated and repeating this argument inductively, $v$ gets activated at the time step $i-(i_j-i_{j'})$ which contradicts $i$ being the first time step when $v$ gets activated. It follows that $q\le |V(G)|\cdot t_{\max}$ and since $i_j - i_{j-1} \le \delta$ and $i_0\le t_{\max}$ it follows that  $(|V(G)|+1)\cdot t_{\max}\cdot \delta$. The theorem follows by trying all transmission schedules than transmit only in the first $t_{\max}$ steps and running the spreading process for $(|V(G)|+1)\cdot t_{\max}\cdot \delta$ steps for each. Each time we check if at least $k$ vertices have been activated.\fi 
\end{proof}
\fi

Notice that if $\delta > t_{\max}$, then every vertex with at least one neighbor in the underlying graph will remain active forever after being activated. Therefore, in this case, both \problemtwo and \problemthree always require budget at most two: transmitting at time step $1$ will make every connected component of $G -\{s\}$ active by time step $|V(G)|\cdot t_{\max}$; then transmit one more time to activate the singletons in $G -\{s\}$. On the other hand, perhaps surprisingly, we will show that the problem is \NP-hard whenever $\delta < t_{\max}$.

\begin{figure}
    \centering
    \includegraphics[width=\textwidth]{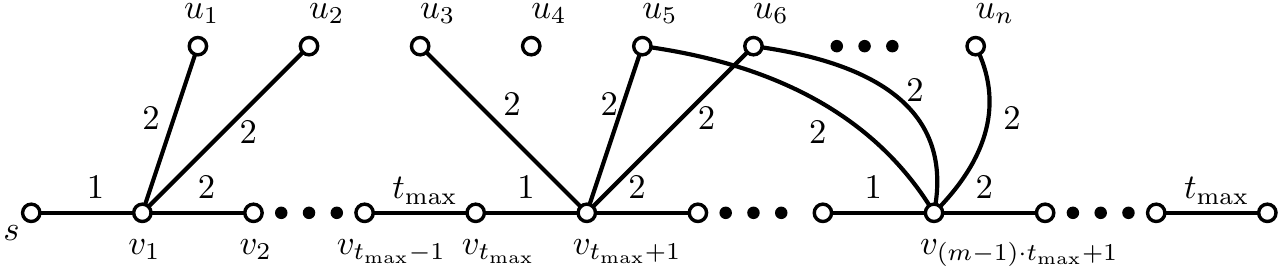}
    \caption{Illustration of the reduction described in Theorem~\ref{thm:PeriodicSpread2}. There is a path $s$, $v_1,\ldots,v_{m\tmax}$ with cycling labels $1,2,\ldots,\tmax,1,2\ldots$ and vertices $u_1,\ldots, u_n$. Vertex $v_1$ is adjacent to $u_1$ and $u_2$ with label $2$ on the two edges representing $S_0 = \{1,2\}$. Similarly, $v_{\tmax+1}$ is adjacent to $u_3,u_5,u_6$ with label $2$, representing $S_1=\{u_3,u_5,u_6\}$ and so on.} 
    \label{fig:periodic_hardness}
\end{figure}

\begin{theorem}\label{thm:PeriodicSpread2}

    \problemtwo and \problemthree are \NP-hard and \Wtwo-hard parameterized by the budget $b$. Moreover, this holds for every pair of $\delta, t_{\max}$, such that $\delta < t_{\max}$, $t_{\max}\ge 2$ and $\delta\ge 1$.
\end{theorem}

\begin{proof}
    Let $\tuple{\scal, N, b}$ be an instance of \setcover such that $\scal = \{S_0, S_1,\ldots, S_{m-1}\}$. Without loss of generality let us assume that $N= [n]$ for some $n\in \mathbb{N}$.
    Given any $\delta, t_{\max}$, such that $\delta < t_{\max}$, $t_{\max}\ge 2$, $\delta\ge 1$, we are going to construct an input $\tuple{\gcal, s, \delta, b, k}$ for \problemtwo and \problemthree such that if $\scal$ admits a set cover of size at most $b$, then there exists a transmission schedule $\tcal$ of size at most $b$ such that
    $\left|(\delta,\tcal)-\activeVertex_t(\gcal,s)\right| \geq k$, where $t=(m-1)\cdot t_{\max}+\delta+2$ and         
    $\max_{t \in [\tmax]}\left|(\delta,\tcal)-\activeVertex_t(\gcal,s)\right| < k$ otherwise. We set $k = n + b\cdot \delta$. 
    
    The temporal graph $\gcal$ is constructed as follows. See also Figure~\ref{fig:periodic_hardness} for the illustration of the reduction.
    The vertex set of $\gcal$ consists of: \ifshort 1) source vertex $s$; 2) $q = m\cdot t_{max}$ many vertices $V_1 = \{v_1, v_2, \ldots, v_q\}$; 3) vertices $V_2 = \{u_1, u_2, \ldots, u_n\}$, and the edges of \gcal are: 1) edge $sv_1$ with label $1$;
        2) for every $i\in [q-1]$, edge $v_{i}v_{i+1}$ with label $(i\mod t_{\max})+1$;
        3) for every set $S_j\in \scal$ and every $i\in S_j$, there is an edge between $v_(j\cdot t_{\max})+1$ and $u_i$ with label $2$.  \fi 
    \iflong
    \begin{itemize}
        \item source vertex $s$;
        \item $q = m\cdot t_{max}$ many vertices $V_1 = \{v_1, v_2, \ldots, v_q\}$;
        \item vertices $V_2 = \{u_1, u_2, \ldots, u_n\}$.
    \end{itemize}
     
    The temporal edges of $\gcal$ are as follows: 
    \begin{itemize}
        \item edge $sv_1$ with label $1$;
        \item for every $i\in [q-1]$, edge $v_{i}v_{i+1}$ with label $(i\mod t_{\max})+1$;
        \item for every set $S_j\in \scal$ and every $i\in S_j$, there is an edge between $v_(j\cdot t_{\max})+1$ and $u_i$ with label $2$. 
    \end{itemize}
    \fi

    \ifshort
        Since  edges $v_ju_i$ have all label $2$ and $\delta<\tmax$, it is easy to see that vertices in $V_2$ can only be activated in time step $\ell\cdot t_{\max} + 3$ for some $\ell\in \mathbb{N}$ and it will be active in time steps $\ell\cdot t_{\max} + 3, \ell\cdot t_{\max} + 4, \ldots, \ell\cdot t_{\max} + \delta + 2$. Similarly, $v_i$ for $i\in [q]$ can only be activated in time step $\ell\cdot t_{\max} + (i\mod t_{\max})+1$ from $v_{i-1}$ (where $v_0=s$). It is also straightforward to verify that we can assume that the source $s$ transmits only at time steps $\tau$ such that $\tau = \ell\cdot \tmax+1$. Now, if $s$ transmits at $\tau = \ell\cdot \tmax+1$, then at time step $t$, only the vertices $v_{t-\tau-\delta}, \ldots, v_{t-\tau-1}$ in $V_1$ are active. Moreover, if $t-\tau = j\cdot \tmax + r$, where $r\in \{2, 3, \ldots, \delta + 1\}$, then at time step $t$, the vertices $u_i\in V_2$ such that $i\in S_j$ are also active at the time step $t$ and these are the only vertices of $V_2$ that are activated at time step $t$ by the transmission at the time step $\tau$.
    \fi
 \iflong    
    Note that edges between $u_i$'s and $v_j$'s have all label $2$. If follows that any of $u_i$'s can only be activated in time step $\ell\cdot t_{\max} + 3$ for some $\ell\in \mathbb{N}$ and it will be active in time steps $\ell\cdot t_{\max} + 3, \ell\cdot t_{\max} + 4, \ldots, \ell\cdot t_{\max} + \delta + 2$. In particular, since $\delta < t_{\max}$, it will not be active in time step $(\ell+1)\cdot t_{\max} + 3$ anymore. Therefore, a vertex $v_j$, $j\in [q]$ will never be activated by $u_i$, $i\in [n]$. 
    Now, note that $(s, v_1, v_2, \ldots, v_q)$ forms a path such that the labels on the two consecutive edges always increase by one up to $t_{\max}$ and the label after $t_{\max}$ is $1$ again. This means that $v_i$ for $i\in [q]$ can only be activated in time step  $\ell\cdot t_{\max} + (i\mod t_{\max})+1$ from $v_{i-1}$ (or $s$ in case of $v_1$) and because $\delta < t_{\max}$ $v_i$ will not be anymore active at time step $(\ell+1)\cdot t_{\max} + (i\mod t_{\max})$, when the edge $v_{i-1}v_i$ is active 
    again. Given the above discussion, it is rather straightforward to verify that for the source $s$ it only make sense to transmit at the time step $\tau$ such that $\tau = \ell\cdot \tmax+1$ for some $\ell\in \mathbb{N}$, else either its counter runs out before spreading to $v_1$ or it only spreads to $v_1$ at the next time step $\tau'$ such that $\tau'\mod \tmax = 1$. Now, if $s$ transmits at $\tau = \ell\cdot \tmax+1$, then at the time step $t$, only the vertices $v_{t-\tau-\delta}, \ldots, v_{t-\tau-1}$ in $V_1$ are active. Moreover, if $t-\tau = j\cdot \tmax + r$, where $r\in \{2, 3, \ldots, \delta + 1\}$, then at the time step $t$, the vertices $u_i\in V_2$ such that $i\in S_j$ are also active at the time step $t$ and these are the only vertices of $V_2$ that are activated at time step $t$ by the transmission at the time step $\tau$. More precisely, if $t-\tau = j\cdot \tmax + r$, where $r\in \{2, 3, \ldots, \delta + 1\}$, then $(\delta,\tau)-\activeVertex_t(\gcal,s) = \{v_{t-\tau-\delta}, \ldots, v_{t-\tau-1}\}\cup \bigcup_{i\in S_j}\{u_i\}$, else $(\delta,\tau)-\activeVertex_t(\gcal,s) = \{v_{t-\tau-\delta}, \ldots, v_{t-\tau-1}\}$. 
\fi 

    Given the above discussion and Lemma~\ref{lem:ApproxMain}, we can quite easily argue that a transmission schedule $\tcal= \{\tau_1, \ldots, \tau_b\}$ that activates at least $k = n + b\cdot \delta$ many vertices at time step $t$, then we can construct a set cover of $\scal$ of size $b$. 
    Namely, it follows that if we transmit at most $b$ times, then at most $b\cdot (\delta - 1)$ many vertices in $V_1$ can be active at the same time and in order for $k = n + b\cdot \delta$ many vertices to be active at time step $t$, all the vertices in $V_2$ have to be active. It follows that $\tcal= \{\tau_1, \ldots, \tau_b\}$ such that 
    for all $i\in [b]$, $t-\tau_i = j_i\cdot \tmax + r_i$, where $r_i\in \{2, 3, \ldots, \delta + 1\}$, and $\bigcup_{i\in [b]}S_{j_i}=[n]$, therefore $\scal$ admits a set cover of size at most $b$. 
    
    %$(m-j)t_{\max} + 1$, then at the time step 
    %$m\tmax+2$, the vertex $v_{j\tmax+1}$ is active, in time step $m\tmax+3$, all the vertices $u_i$ such that $i\in S_j$ and the vertex $v_{jt_{\max}+2}$ get 
    %active, and by the time $m\tmax+\delta + 2$, the vertices  $u_i$ such that $i\in S_j$ and the vertices $v_{j\tmax+2}$, $v_{j\tmax+3}$, $\ldots$ 
    %,$v_{j\tmax+\delta+1}$ are active. 
\ifshort
On the other hand, we can easily revert the above construction and postpone the first transmission such that the active vertices on the path are in the blocks representing set cover exactly at time step  $t = (m-1)\cdot \tmax+\delta+2$, at this time step all vertices in $U$ activated by the first vertex in these blocks are still active.
\fi
\iflong  
    On the other hand, the above discussion implies that if $t=(m-1)\cdot \tmax+\delta+2$, then for every $\tau$ such that $\tau = (m-j)\tmax + 1$ 
    it holds that $(\delta,\tau)-\activeVertex_t(\gcal,s) = \{v_{j\tmax+2},v_{j\tmax+3},\ldots,v_{j\tmax+\delta+1}\}\bigcup_{i\in S_j}\{u_i\}$. Hence, it follows that if $S_{j_1}, \ldots, S_{j_b}$ is a set cover, then by $\tcal = \{ (m-j_b)\tmax + 1, (m-j_{b-1})\tmax + 1, (m-j_1)\tmax + 1 \}$ is a transmission schedule that actives all the vertices in $V_2$ and $b\cdot \delta$ many vertices in $V_1$ at the time $t$, that is  $\left|(\delta,\tcal)-\activeVertex_t(\gcal,s)\right| = k$. Hence, if $\scal$ admits a set cover of size at most $b$, then we also have a transmission schedule such that $\left|(\delta,\tcal)-\activeVertex_t(\gcal,s)\right| \ge k$ satisfying both \problemthree at time step $t = (m-1)\cdot \tmax+\delta+2$ and \problemtwo.

    Note that this reduction is polynomial time reduction and the size of set cover translates to the size of the transmission schedule. Hence the statement of the theorem follows from the fact that \setcover is \NP-hard and \Wtwo-hard by the size of the sought set cover.
\fi
\end{proof}
%%%%%%%%%%%%%%%%%%%%%%%%%%%%%%%%%%%%%%%%%%%%%%%%%%%%%%%%%%%%%%%%%%%%%%%%
\section{Discussion}
\label{sec:discuss}

In this paper we have explored the complexity of influence maximization on temporal graphs with a single fixed source. We have focused on four objectives, \problemone, \problemtwo, \problemthree, and \problemfour under three different settings, which are naturally motivated by real life scenarios. We have proved that in almost every case, the problem is intractable. 

In this section we discuss the connections our model has with some other problems arising in temporal graphs; we compare our spreading dynamics with the ``standard'' SIS model; and we highlight some open questions that deserve extra study.

\medskip
\noindent {\bf Comparison with the SIS model.} As we have explained in Section~\ref{sec:prelims}, the spreading dynamics we consider allow, ``renewal'' of the influence. In other words, a vertex will reset its counter to $\delta$ every time it is adjacent to an active vertex, even if it is {\em already} active. In contrast, the original SIS model does not allow this; a vertex $u$ becomes active at time step $t+1$ only if at time step $t$, vertex $u$ is (a) inactive and (b) adjacent to an active vertex. 
This difference makes the two processes behave very differently. In fact, under the SIS model, the size of active vertices does not monotonically increase with the number of transmissions; Figure~\ref{fig:sis-comparison} demonstrates this.
\begin{figure}
    \centering
    \includegraphics[width=0.90\textwidth]{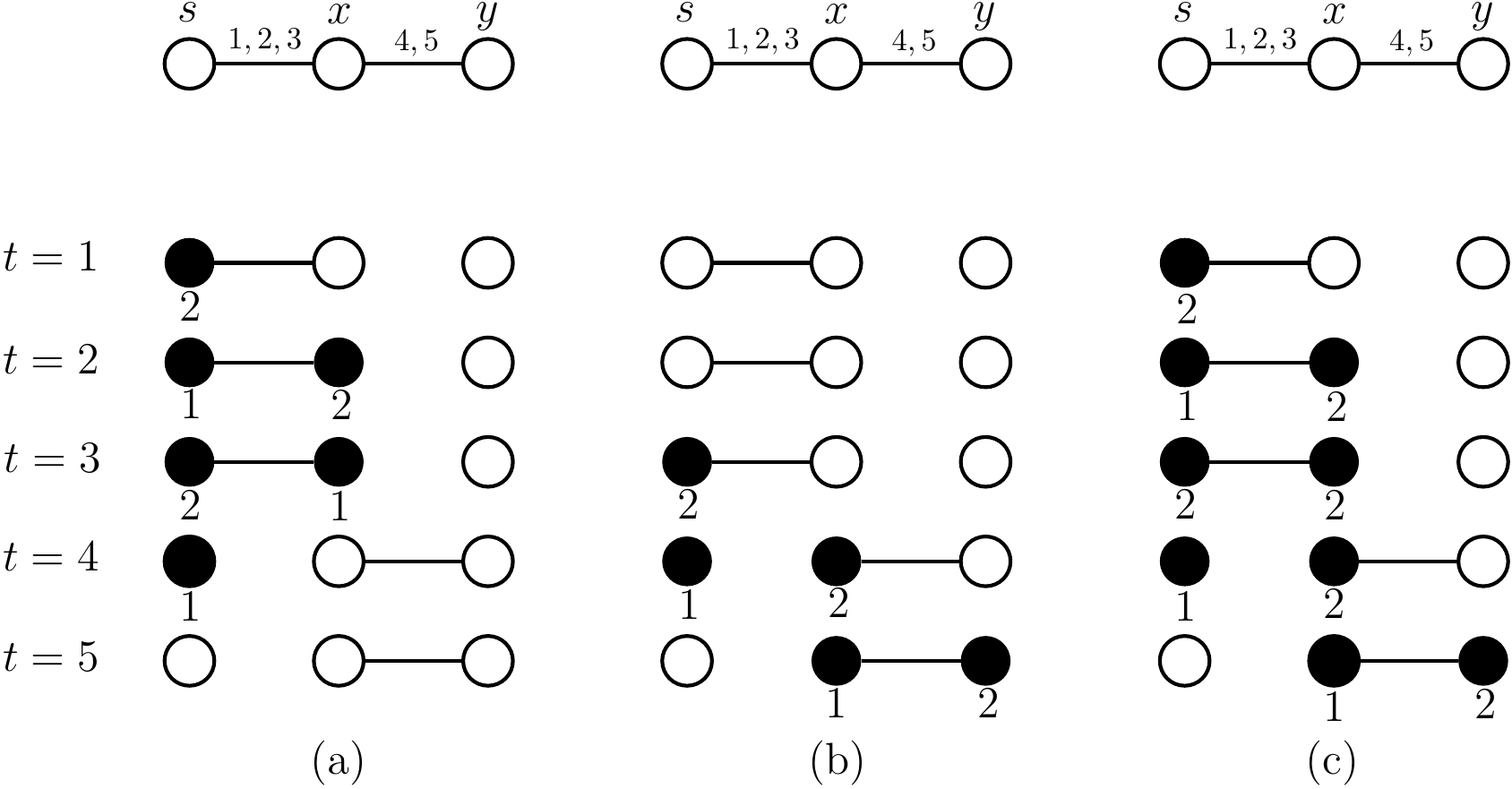}
    \caption{Comparison between transmission schedules for SIS and our spreading dynamics where $\delta=2$, on the temporal path with vertices $s, x, y$. (a) SIS dynamics with (a) $\tcal= (1,3)$; (b) SIS dynamics with $\tcal= (3)$; (c) our dynamics with $\tcal= (1,3)$. Solid vertices depict active vertices. The number below each vertex is its counter at the specific timestep.}
    \label{fig:sis-comparison}
\end{figure}
\noindent
It is not too difficult to verify that all of our hardness results, both \NP-hardness and \Wtwo-hardness, apply under the SIS spreading model. This is because the structure of our instances is designed in a way that does not allow for renewal. On the other hand though, our positive result for periodic graphs is no longer valid under SIS model. The complexity of the problem is an interesting question, mainly because of the non-monotonicity of the active vertices when SIS is used.

\medskip
\noindent {\bf Open Problem \showmycounter.} What is the complexity of \problemone for periodic graphs under the SIS spreading model?

\medskip
\noindent {\bf Connection to restless temporal walks.}
A {\em temporal walk} in $\tuple{G, \ecal}$ from vertex $v_1$ to vertex $v_w$ is a sequence of edges $W = (v_iv_{i+1}, t_i)_{i=1}^{w-1}$ such that for every $i \in [w]$ it holds that $v_iv_{i+1} \in E_{t_i}$, i.e. $v_iv_{i+1}$ is available at time step $t_i$ and time steps are strictly increasing, i.e. if $i < j$ then $t_i < t_j$.  Observe that a vertex $v$ can appear multiple times in a temporal walk; in a temporal path this is {\em not} allowed. 
A {\em temporal walk} $W = (v_iv_{i+1}, t_i)_{i=1}^{w-1}$ is called $\delta$-restless if $t_i<t_{i+1}<t_i+\delta$, for every $i \in [w]$.
Then observe that $(\delta, \tau)$-$\activeVertex_t(\gcal,s)$ is equal to the set of vertices for which there exists a $(\delta, \tau)$-restless temporal walk from $s$ with arrival time between $t$ and $t-\delta$. 
Restless walks have been studied in the past~\cite{bentert2020efficient}, albeit  from a different point of view compared to ours.

\medskip
\noindent {\bf Open questions.}
We have resolved the complexity of the four objectives for almost every class of graphs. Though, there exist some intriguing questions that will complete the complexity-landscape of the problem. 

\medskip
\noindent {\bf Open Problem \showmycounter.} Do \problemtwo and \problemthree on periodic graphs belong to \NP, or are they complete for some other class, like \PSPACE? 

Observe that the infinite lifetime of the graph makes the problem of verifying objectives \problemtwo and \problemthree non trivial. An intermediate question is the following, since our hardness reduction cannot be trivially extended in order to resolve it.

\medskip
\noindent {\bf Open Problem \showmycounter.} Are \problemtwo and \problemthree on periodic graphs \NP-hard when $\delta = \tmax$, i.e. when $\delta$ equals the period of the graph?

The next question is concerned about the case where the underlying graph is a path. Some initial observations show that \problemone in the unconstrained setting is tractable. However, for the remaining combinations of objectives and transmission schedules, the problem remains wide open.

\medskip
\noindent {\bf Open Problem \showmycounter.} What is the complexity of the four objectives when the underlying graph is a path?

\medskip

The last question concerns with studying \problemfour in the last two transmission schedules.

\medskip
\noindent {\bf Open Problem \showmycounter.} What is the complexity of \problemone for window constrained schedules and periodic graphs?

%%%%%%%%%%%%%%%%%%%%%%%%%%%%%%%%%%%%%%%%%%%%%%%%%%%%%%%%%%%%%%%%%%%%%%%%

%%% The acknowledgments section is defined using the "acks" environment
%%% (rather than an unnumbered section). The use of this environment 
%%% ensures the proper identification of the section in the article 
%%% metadata as well as the consistent spelling of the heading.

% \begin{acks}
% If you wish to include any acknowledgments in your paper (e.g., to 
% people or funding agencies), please do so using the `\texttt{acks}' 
% environment. Note that the text of your acknowledgments will be omitted
% if you compile your document with the `\texttt{anonymous}' option.
% \end{acks}

%%%%%%%%%%%%%%%%%%%%%%%%%%%%%%%%%%%%%%%%%%%%%%%%%%%%%%%%%%%%%%%%%%%%%%%%

%%% The next two lines define, first, the bibliography style to be 
%%% applied, and, second, the bibliography file to be used.
\newpage

\bibliographystyle{plainurl} 
\bibliography{references}

\begin{thebibliography}{10}

\bibitem{aggarwal2012influential}
Charu~C Aggarwal, Shuyang Lin, and Philip~S Yu.
\newblock On influential node discovery in dynamic social networks.
\newblock In {\em Proceedings of the 2012 SIAM International Conference on Data
  Mining}, pages 636--647. SIAM, 2012.

\bibitem{ahn2004competitive}
Hee-Kap Ahn, Siu-Wing Cheng, Otfried Cheong, Mordecai Golin, and Rene
  Van~Oostrum.
\newblock Competitive facility location: the {V}oronoi game.
\newblock {\em Theoretical Computer Science}, 310(1):457--467, 2004.

\bibitem{alon2010note}
Noga Alon, Michal Feldman, Ariel~D Procaccia, and Moshe Tennenholtz.
\newblock A note on competitive diffusion through social networks.
\newblock {\em Information Processing Letters}, 110(6):221--225, 2010.

\bibitem{arora2017debunking-survey2}
Akhil Arora, Sainyam Galhotra, and Sayan Ranu.
\newblock Debunking the myths of influence maximization: An in-depth
  benchmarking study.
\newblock In {\em Proceedings of the 2017 ACM international conference on
  management of data}, pages 651--666, 2017.

\bibitem{bentert2020efficient}
Matthias Bentert, Anne-Sophie Himmel, Andr{\'e} Nichterlein, and Rolf
  Niedermeier.
\newblock Efficient computation of optimal temporal walks under waiting-time
  constraints.
\newblock {\em Applied Network Science}, 5(1):1--26, 2020.

\bibitem{ijcai2021p7}
Niclas Boehmer, Vincent Froese, Julia Henkel, Yvonne Lasars, Rolf Niedermeier,
  and Malte Renken.
\newblock Two influence maximization games on graphs made temporal.
\newblock In {\em Proceedings of the Thirtieth International Joint Conference
  on Artificial Intelligence, {IJCAI-21}}, pages 45--51, 2021.
\newblock \href {https://doi.org/10.24963/ijcai.2021/7}
  {\path{doi:10.24963/ijcai.2021/7}}.

\bibitem{chen2013information-survey11}
Wei Chen, Laks~VS Lakshmanan, and Carlos Castillo.
\newblock Information and influence propagation in social networks.
\newblock {\em Synthesis Lectures on Data Management}, 5(4):1--177, 2013.

\bibitem{chen2012time-temp14}
Wei Chen, Wei Lu, and Ning Zhang.
\newblock Time-critical influence maximization in social networks with
  time-delayed diffusion process.
\newblock In {\em Twenty-Sixth AAAI Conference on Artificial Intelligence},
  page 592–598, 2012.

\bibitem{CyganFKLMPPS15}
Marek Cygan, Fedor~V. Fomin, Lukasz Kowalik, Daniel Lokshtanov, D{\'{a}}niel
  Marx, Marcin Pilipczuk, Michal Pilipczuk, and Saket Saurabh.
\newblock {\em Parameterized Algorithms}.
\newblock Springer, 2015.
\newblock \href {https://doi.org/10.1007/978-3-319-21275-3}
  {\path{doi:10.1007/978-3-319-21275-3}}.

\bibitem{downey1995fixed}
Rod~G Downey and Michael~R Fellows.
\newblock Fixed-parameter tractability and completeness i: Basic results.
\newblock {\em SIAM Journal on computing}, 24(4):873--921, 1995.

\bibitem{DowneyFellows13}
Rodney~G. Downey and Michael Fellows.
\newblock {\em Fundamentals of parameterized complexity}.
\newblock Texts in Computer Science. Springer, 2013.

\bibitem{durr2007nash}
Christoph D{\"u}rr and Nguyen~Kim Thang.
\newblock Nash equilibria in {V}oronoi games on graphs.
\newblock In {\em European Symposium on Algorithms}, pages 17--28. Springer,
  2007.

\bibitem{erkol2020influence}
{\c{S}}irag Erkol, Dario Mazzilli, and Filippo Radicchi.
\newblock Influence maximization on temporal networks.
\newblock {\em Physical Review E}, 102(4):042307, 2020.

\bibitem{erkol2022effective}
{\c{S}}irag Erkol, Dario Mazzilli, and Filippo Radicchi.
\newblock Effective submodularity of influence maximization on temporal
  networks.
\newblock {\em Physical Review E}, 106(3):034301, 2022.

\bibitem{fukuzono2020two}
Naoka Fukuzono, Tesshu Hanaka, Hironori Kiya, Hirotaka Ono, and Ryogo
  Yamaguchi.
\newblock Two-player competitive diffusion game: Graph classes and the
  existence of a {N}ash equilibrium.
\newblock In {\em International Conference on Current Trends in Theory and
  Practice of Informatics}, pages 627--635. Springer, 2020.

\bibitem{GareyJ79}
M.~R. Garey and David~S. Johnson.
\newblock {\em Computers and Intractability: {A} Guide to the Theory of
  NP-Completeness}.
\newblock 1979.

\bibitem{gayraud2015diffusion}
Nathalie~TH Gayraud, Evaggelia Pitoura, and Panayiotis Tsaparas.
\newblock Diffusion maximization in evolving social networks.
\newblock In {\em Proceedings of the 2015 acm on conference on online social
  networks}, pages 125--135, 2015.

\bibitem{goldenberg2001talk}
Jacob Goldenberg, Barak Libai, and Eitan Muller.
\newblock Talk of the network: A complex systems look at the underlying process
  of word-of-mouth.
\newblock {\em Marketing letters}, 12(3):211--223, 2001.

\bibitem{goldenberg2001using}
Jacob Goldenberg, Barak Libai, and Eitan Muller.
\newblock Using complex systems analysis to advance marketing theory
  development: Modeling heterogeneity effects on new product growth through
  stochastic cellular automata.
\newblock {\em Academy of Marketing Science Review}, 9(3):1--18, 2001.

\bibitem{guille2013information-survey40}
Adrien Guille, Hakim Hacid, Cecile Favre, and Djamel~A Zighed.
\newblock Information diffusion in online social networks: A survey.
\newblock {\em ACM Sigmod Record}, 42(2):17--28, 2013.

\bibitem{hochba1997approximation}
Dorit~S Hochba.
\newblock Approximation algorithms for np-hard problems.
\newblock {\em ACM Sigact News}, 28(2):40--52, 1997.

\bibitem{holme2012temporal}
Petter Holme and Jari Saram{\"a}ki.
\newblock Temporal networks.
\newblock {\em Physics reports}, 519(3):97--125, 2012.

\bibitem{kanuri2018scheduling}
Vamsi Kanuri, Yixing Chen, and Shrihari Sridhar.
\newblock Scheduling content on social media: Theory, evidence, and
  application.
\newblock {\em Journal of Marketing}, 86:89--108, 11 2018.
\newblock \href {https://doi.org/10.1177/0022242918805411}
  {\path{doi:10.1177/0022242918805411}}.

\bibitem{kempe02-temporal}
David Kempe, Jon Kleinberg, and Amit Kumar.
\newblock Connectivity and inference problems for temporal networks.
\newblock {\em Journal of Computer and System Sciences}, 64(4):820--842, 2002.

\bibitem{kempe2003maximizing}
David Kempe, Jon Kleinberg, and {\'E}va Tardos.
\newblock Maximizing the spread of influence through a social network.
\newblock In {\em Proceedings of the ninth ACM SIGKDD international conference
  on Knowledge discovery and data mining}, pages 137--146, 2003.

\bibitem{kim2014ct-temp60}
Jinha Kim, Wonyeol Lee, and Hwanjo Yu.
\newblock Ct-ic: Continuously activated and time-restricted independent cascade
  model for viral marketing.
\newblock {\em Knowledge-Based Systems}, 62:57--68, 2014.

\bibitem{li2018influence-survey}
Yuchen Li, Ju~Fan, Yanhao Wang, and Kian-Lee Tan.
\newblock Influence maximization on social graphs: A survey.
\newblock {\em IEEE Transactions on Knowledge and Data Engineering},
  30(10):1852--1872, 2018.

\bibitem{lin2015learning-temp71}
Su-Chen Lin, Shou-De Lin, and Ming-Syan Chen.
\newblock A learning-based framework to handle multi-round multi-party
  influence maximization on social networks.
\newblock In {\em Proceedings of the 21th ACM SIGKDD International Conference
  on Knowledge Discovery and Data Mining}, pages 695--704, 2015.

\bibitem{rodriguez2011uncovering-temp90}
Manuel~Gomez Rodriguez, David Balduzzi, and Bernhard Sch{\"o}lkopf.
\newblock Uncovering the temporal dynamics of diffusion networks.
\newblock {\em arXiv preprint arXiv:1105.0697}, 2011.

\bibitem{scaman2015anytime-temp93}
Kevin Scaman, R{\'e}mi Lemonnier, and Nicolas Vayatis.
\newblock Anytime influence bounds and the explosive behavior of
  continuous-time diffusion networks.
\newblock {\em Advances in Neural Information Processing Systems}, 28, 2015.

\bibitem{spasojevic2015post}
Nemanja Spasojevic, Zhisheng Li, Adithya Rao, and Prantik Bhattacharyya.
\newblock When-to-post on social networks.
\newblock In {\em Proceedings of the 21th ACM SIGKDD International Conference
  on Knowledge Discovery and Data Mining}, pages 2127--2136, 2015.

\bibitem{sun2011-survey96}
Jimeng Sun and Jie Tang.
\newblock A survey of models and algorithms for social influence analysis.
\newblock In {\em Social network data analytics}, pages 177--214. Springer,
  2011.

\bibitem{tejaswi2016diffusion-survey101}
V~Tejaswi, PV~Bindu, and P~Santhi Thilagam.
\newblock Diffusion models and approaches for influence maximization in social
  networks.
\newblock In {\em 2016 International Conference on Advances in Computing,
  Communications and Informatics (ICACCI)}, pages 1345--1351. IEEE, 2016.

\bibitem{wang2018modeling}
Wenjun Wang and W~Nick Street.
\newblock Modeling and maximizing influence diffusion in social networks for
  viral marketing.
\newblock {\em Applied network science}, 3(1):1--26, 2018.

\bibitem{xie2015dynadiffuse-temp111}
Miao Xie, Qiusong Yang, Qing Wang, Gao Cong, and Gerard De~Melo.
\newblock Dynadiffuse: A dynamic diffusion model for continuous time
  constrained influence maximization.
\newblock In {\em Twenty-Ninth AAAI Conference on Artificial Intelligence},
  2015.

\bibitem{zarezade2017redqueen}
Ali Zarezade, Utkarsh Upadhyay, Hamid~R. Rabiee, and Manuel Gomez-Rodriguez.
\newblock Redqueen: An online algorithm for smart broadcasting in social
  networks.
\newblock WSDM '17, page 51–60, New York, NY, USA, 2017. Association for
  Computing Machinery.
\newblock \href {https://doi.org/10.1145/3018661.3018684}
  {\path{doi:10.1145/3018661.3018684}}.

\bibitem{zhang2013maximizing}
Huiyuan Zhang, Thang~N Dinh, and My~T Thai.
\newblock Maximizing the spread of positive influence in online social
  networks.
\newblock In {\em 2013 IEEE 33rd International Conference on Distributed
  Computing Systems}, pages 317--326. IEEE, 2013.

\bibitem{zhang2014recent-survey113}
Huiyuan Zhang, Subhankar Mishra, My~T Thai, J~Wu, and Y~Wang.
\newblock Recent advances in information diffusion and influence maximization
  in complex social networks.
\newblock {\em Opportunistic Mobile Social Networks}, 37(1.1):37, 2014.

\bibitem{zhou2015location}
Tao Zhou, Jiuxin Cao, Bo~Liu, Shuai Xu, Ziqing Zhu, and Junzhou Luo.
\newblock Location-based influence maximization in social networks.
\newblock In {\em Proceedings of the 24th ACM international on conference on
  information and knowledge management}, pages 1211--1220, 2015.

\bibitem{zhu2014maximizing}
Tian Zhu, Bai Wang, Bin Wu, and Chuanxi Zhu.
\newblock Maximizing the spread of influence ranking in social networks.
\newblock {\em Information Sciences}, 278:535--544, 2014.

\end{thebibliography}

%%%%%%%%%%%%%%%%%%%%%%%%%%%%%%%%%%%%%%%%%%%%%%%%%%%%%%%%%%%%%%%%%%%%%%%%

\end{document}